\documentclass[11pt, oneside]{amsart}
\usepackage{amsaddr}

\title{Optimal Las Vegas Approximate Near Neighbors in $\ell_p$}
\author{Alexander Wei}
\address{Harvard University}
\email{weia@college.harvard.edu}
\thanks{Supported by a Harvard PRISE Fellowship and a Herchel Smith Fellowship.}

\usepackage[utf8x]{inputenc}
\usepackage[T1]{fontenc}
\usepackage{lmodern}
\usepackage[margin=1in]{geometry}

\usepackage{enumerate}
\usepackage{varioref}
\usepackage{hyperref}
\usepackage{cleveref}

\usepackage{color}
\usepackage{tikz-cd}
\usepackage{mathtools}

\usepackage[nosetup]{awei}

\hypersetup{
  linkcolor  = blue!80!black,
  citecolor  = blue!80!black,
  colorlinks = true,
}

\nc{\C}{c}
\nc{\N}{n}
\nc{\dimA}{d}
\nc{\dimB}{B}
\nc{\dimC}{b}
\nc{\expC}{\alpha}
\nc{\expB}{\beta}
\nc{\dist}[2]{\norm{#1 - #2}_2}
\nc{\Filt}{\mathfrak{F}}
\nc{\FiltB}{\mathfrak{G}}
\nc{\Proj}{\mathfrak{P}}
\nc{\Ball}{\cB}
\nc{\Sphere}{\cS}
\nc{\SphCap}{W}
\nc{\frad}{w}
\nc{\fnum}{N}
\nc{\fdel}{\delta}
\nc{\epsB}{\eps}
\nc{\epsA}{\eps}
\nc{\G}{G}
\nc{\F}{F}
\nc{\K}{K}
\nc{\Alpha}{\mathfrak{c}}
\nc{\Beta}{\mathfrak{s}}

\nc{\todo}[1]{{\color{red}[#1]}}
\nc{\block}[1]{
  \underbrace{\begin{matrix} \sqrt{\frac{d'}{d}} & \sqrt{\frac{d'}{d}} & \cdots & \sqrt{\frac{d'}{d}}\end{matrix}}_{#1}
}

\begin{document}

\begin{abstract}
  We show that approximate near neighbor search in high dimensions can be solved in a Las Vegas fashion (i.e., without false negatives) for $\ell_p$ ($1\le p\le 2$) while matching the performance of optimal locality-sensitive hashing. Specifically, we construct a data-independent Las Vegas data structure with query time $O(dn^{\rho})$ and space usage $O(dn^{1+\rho})$ for $(r, c r)$-approximate near neighbors in $\mathbb{R}^{d}$ under the $\ell_p$ norm, where $\rho = 1/c^p + o(1)$. Furthermore, we give a Las Vegas locality-sensitive filter construction for the unit sphere that can be used with the data-dependent data structure of Andoni et al. (SODA 2017) to achieve optimal space-time tradeoffs in the data-dependent setting. For the symmetric case, this gives us a data-dependent Las Vegas data structure with query time $O(dn^{\rho})$ and space usage $O(dn^{1+\rho})$ for $(r, c r)$-approximate near neighbors in $\mathbb{R}^{d}$ under the $\ell_p$ norm, where $\rho = 1/(2c^p - 1) + o(1)$.

  Our data-independent construction improves on the recent Las Vegas data structure of Ahle (FOCS 2017) for $\ell_p$ when $1 < p\le 2$. Our data-dependent construction does even better for $\ell_p$ for all $p\in [1, 2]$ and is the first Las Vegas approximate near neighbors data structure to make use of data-dependent approaches. We also answer open questions of Indyk (SODA 2000), Pagh (SODA 2016), and Ahle by showing that for approximate near neighbors, Las Vegas data structures can match state-of-the-art Monte Carlo data structures in performance for both the data-independent and data-dependent settings and across space-time tradeoffs.
\end{abstract}

\maketitle
\setcounter{page}{0}
\thispagestyle{empty}
\newpage

\section{Introduction}

In the nearest neighbor search problem, we are given a database $\cD$ of $\N$ points in a metric space $(X, D)$ and asked to construct a data structure to efficiently answer queries for the point in $\cD$ that is closest to a query point $q\in X$. Nearest neighbor search has found applications to a diverse set of areas in computer science, including machine learning, computer vision, and document retrieval. For many of these applications, the space $X$ is $\RR^\dimA$ and the metric $D$ is an $\ell_p$ norm.

The nearest neighbor search problem can be solved efficiently when the dimension of the space $X$ is ``low'' (e.g., \cite{Clarkson88, Meiser93}). However, all known data structures for nearest neighbor search suffer from the ``curse of dimensionality,'' in which either the space or time complexity of the data structure is exponential in the dimension $\dimA$. Thus researchers (e.g., \cite{HIM12, GIM99, KOR00}) have also considered an approximate version of nearest neighbor search, stated as follows:
\begin{defn}[$(r, \C r)$-approximate near neighbors]
  Given a database $\cD$ of $\N$ points in a metric space $(X, D)$, construct a data structure which, given a query point $q\in X$, returns a point in $\cD$ within distance $\C r$ of $q$, provided there exists a point in $\cD$ within distance $r$ of $q$.
\end{defn}

In their seminal work, \cite{HIM12} propose \emph{locality-sensitive hashing} (LSH) as a solution for approximate near neighbors and provide an LSH data structure for Hamming space that overcomes the curse of dimensionality. In particular, their data structure obtains query time $O(\dimA n^\rho)$ and space usage $O(\dimA n^{1 + \rho})$ for an exponent of $\rho = 1 / \C$. Over the past two decades, the LSH framework and its generalization \emph{locality-sensitive filters} (LSF) \cite{BDGL16, Christiani17, ALRW17} have emerged as leading approaches for solving approximate near neighbors problems, both in theory and in practice. These frameworks serve as the basis of efficient approximate near neighbors data structures for a variety of metric spaces (e.g., \cite{HIM12, BGMZ97, DIIM04, AI06, CP17}; see also \cite{AI17, AIR18} for surveys). Corresponding lower bounds for the LSH framework have also been found \cite{MNP07, OWZ14}, with the construction of \cite{HIM12} known to be optimal for Hamming space via the result of \cite{OWZ14}.

A more recent line of work \cite{AINR14, AR15, ALRW17} has been on \emph{data-dependent} LSH, approaches towards solving approximate near neighbors in which the ``bucketing'' of points depends on the database $\cD$. This stands in contrast to earlier \emph{data-independent} LSH and LSF (e.g., as in \cite{HIM12}), in which the families of hash functions and filters used are determined independently of $\cD$. The data dependence allows such data structures to overcome the lower bound of \cite{OWZ14} for data-independent LSH, with \cite{AR15} obtaining an improved exponent of $\rho = 1 / (2\C^2 - 1) + o(1)$ for Euclidean approximate near neighbors. More recently, \cite{ALRW17} obtain space-time tradeoffs with data-dependent hashing, extending the result of \cite{AR15} to different exponents for space usage and query time.

Almost all data structures that solve approximate near neighbors while overcoming the curse of dimensionality are Monte Carlo. That is, they fail with some probability to find a near neighbor of the query point when one exists. There are applications, however, where such a failure probability is undesirable, e.g., fraud detection or fingerprint lookups within a criminal database. Furthermore, tuning the failure probability precisely can be difficult in practical settings \cite{GIM99}. These shortcomings of the Monte Carlo approach have motivated the study of Las Vegas data structures for approximate near neighbor search, in which the data structure must \emph{always} return a near neighbor when one exists, but the runtime of the queries is permitted to be a random variable.\footnote{Note that a Las Vegas data structure can be converted into a Monte Carlo data structure with failure probability $\delta$ by running it for $\log(1/\delta)$ times the expected runtime before breaking if no near neighbor is returned. However, no reduction in the other direction (from Monte Carlo to Las Vegas) is known to exist.} In the literature, this problem has also been referred to as that of constructing approximate near neighbors data structures that are ``without false negatives'' \cite{Pagh16, GPSS17}, ``with total recall'' \cite{PP16}, ``exact'' \cite{AGK06}, or ``explicit'' \cite{KKKC16}.

The problem of constructing Las Vegas data structures for approximate near neighbors traces back to \cite{Indyk00} and \cite{AGK06}. However, it was not until recently that Las Vegas data structures have come close to matching the bounds achieved by the best Monte Carlo LSH data structures: \cite{Pagh16} constructs a Hamming space data structure with exponent $\rho = O(1 / \C)$, coming within a constant factor of the optimal $\rho = 1 / \C$ of \cite{HIM12}, and \cite{Ahle17} closes the gap, achieving $\rho = 1 / \C + o(1)$ with a new construction based on dimensionality reductions and locality-sensitive filters. By standard reductions (see \cite{HIM12}), the Hamming space data structure of \cite{Ahle17} can be used to solve $(r, \C r)$-approximate neighbors for $\ell_p$ with $\rho = 1 / \C + o(1)$ for all $p\in [1, 2]$. Although this result is optimal for Hamming space and $\ell_1$ in the data-independent setting, it leaves open whether a Las Vegas data structure can match the performance of Monte Carlo data structures for $p\in (1, 2]$ and in particular for the important case of Euclidean space (i.e., $p = 2$). Progress on this front was made by \cite{SW17, Wygocki17}, which give new Las Vegas data structures for Euclidean $(r, \C r)$-approximate near neighbors. However, these data structures do not match the performance of the best Monte Carlo approaches for this problem and in some instances require exponential space or an approximation factor that is $\omega(1)$.

In this paper, we resolve the above open problem from \cite{Indyk00, Pagh16, Ahle17} by constructing the first Las Vegas data structure for Euclidean $(r, \C r)$-approximate near neighbors with exponent $\rho = 1 / \C^2 + o(1)$. This exponent matches that of the ball lattice LSH construction of \cite{AI06}, which was later shown to be optimal in the data-independent setting by \cite{OWZ14}. By a reduction \cite[Section 5]{Nguyen14}, our data structure implies a data structure for approximate near neighbors in $\ell_p$ for $p\in (1, 2)$ with exponent $\rho = 1 / \C^p + o(1)$, which is again tight by \cite{OWZ14}.\footnote{Our data structure also implies a new Las Vegas construction for approximate near neighbors in Hamming space and $\ell_1$, matching the bounds of \cite{Ahle17}. This can be done via the embedding of Hamming space into Euclidean space (with distances squared) and the reduction from $\ell_1$ to Hamming space, respectively.}

We achieve this result by combining the approaches of \cite{Ahle17} and \cite{AI06} with some new applications of dimensionality reduction to derandomizing geometric problems. In particular, we modify and extend the general approach outlined in \cite{Ahle17} for Las Vegas data structures for approximate near neighbors to the more difficult $\ell_2$ case. Furthermore, we translate ball lattice hashing \cite{AI06} to the LSF framework, construct using CountSketch \cite{CCF04} a geometric analog of the \emph{splitters} described in \cite{NSS95, AMS06}, and give a two-stage sequence of dimensionality reductions that allows for more efficient processing of false positives than in \cite{Ahle17}. These techniques culminate in the following theorem and corollary:

\begin{thm}\label{thm:main}
  There exists a Las Vegas data structure for Euclidean $(r, \C r)$-approximate near neighbors in $\RR^\dimA$ with expected query time $\tilde O(\dimA\N^{\rho})$, space usage $\tilde O(\dimA\N^{1 + \rho})$, and preprocessing time $\tilde O(\dimA\N^{1 + \rho})$ for $\rho = 1 / \C^2 + o(1)$.
\end{thm}

\begin{cor}\label{cor:lp}
  There exists a Las Vegas data structure for $(r, \C r)$-approximate near neighbors in $\RR^\dimA$ under $\ell_p$ for $1 < p < 2$ with expected query time $\tilde O(\dimA\N^{\rho})$, space usage $\tilde O(\dimA\N^{1 + \rho})$, and preprocessing time $\tilde O(\dimA\N^{1 + \rho})$ for $\rho = 1 / \C^p + o(1)$.
\end{cor}

For the data-dependent setting, we apply similar high-level techniques to get a Las Vegas version of the data-independent spherical LSF used in the data-dependent data structure of \cite{ALRW17}. We then substitute our Las Vegas LSF family in for the spherical LSF family of \cite{ALRW17} to obtain a Las Vegas data structure for data-dependent hashing with space-time tradeoffs:

\begin{thm}\label{thm:SA}
  Let $p$ be such that $1\le p\le 2$, and suppose $\rho_u, \rho_q\ge 0$ are such that
  \[ \C^p\sqrt{\rho_q} + (\C^p - 1)\sqrt{\rho_u}\ge\sqrt{2\C^p - 1}. \]
  There exists a Las Vegas $(r, \C r)$-approximate near neighbors data structure in $\RR^\dimA$ under $\ell_p$ with expected query time $\tilde O(\dimA\N^{\rho_q + o(1)})$, space usage $\tilde O(\dimA\N^{1 + \rho_u + o(1)})$, and preprocessing time $\tilde O(\dimA\N^{1 + \rho_u + o(1)})$.
\end{thm}

This result matches optimal data-dependent hashing bounds given by \cite{AR15, AR16, ALRW17} and resolves open questions of \cite{Ahle17} regarding data-dependence and space-time tradeoffs for Las Vegas data structures. In particular, we have an exponent of $\rho = 1 / (2\C^2 - 1) + o(1)$ for Euclidean approximate near neighbors in the symmetric case $\rho_u = \rho_q$. Since the approach we take to make the data-independent spherical LSH of \cite{ALRW17} into a Las Vegas filter family applies similar ideas and techniques as those of \Cref{thm:main} and \Cref{cor:lp}, but in a more technically involved way, we defer further discussion of this construction to \Cref{appendix:B}.

The combination of \Cref{thm:main}, \Cref{cor:lp}, and \Cref{thm:SA} shows that relative to Monte Carlo data structures, Las Vegas data structures for approximate near neighbors in $\RR^\dimA$ under the $\ell_p$ norm no longer have ``polynomially higher'' query time as noted in the survey \cite{AI17}, but rather match the best Monte Carlo constructions for all $p\in [1, 2]$ in both the data-independent and data-dependent settings.

\subsection{Background and Techniques}\label{subsec:techniques}

LSH \cite{HIM12} has been one of the most popular foundations for approximate near neighbors data structures in recent years. The basic object of this framework is a \emph{locality-sensitive hash family}, a distribution over hash functions such that any pair of ``close'' points is mapped to the same bucket with probability \emph{at least} $p_1$ and any pair of ``distant'' points is mapped to the same bucket with probability \emph{at most} $p_2$. A major contribution of \cite{HIM12} is the result that such a hash family implies a data structure for approximate near neighbors with exponent $\rho = \log(1 / p_1) / \log(1 / p_2)$.

Our construction roughly fits into the \emph{locality-sensitive filters} (LSF) framework, a generalization of LSH to a setting where each point is associated to a \emph{set} of filters instead of a single hash bucket. The basic object of this framework is a \emph{locality-sensitive filter family} (see \Cref{sec:preliminaries}). This framework allows for greater flexibility on the algorithm designer's part \cite{BDGL16, Christiani17, ALRW17} and has been successfully used to construct Las Vegas locality-sensitive data structures for Hamming space \cite{Ahle17}. There are also known lower bounds for LSF: \cite{Christiani17} extends the lower bound of \cite{OWZ14} to a Monte Carlo formalization of LSF, and in this paper, we extend the results of \cite{OWZ14, Christiani17} further to LSF families with Las Vegas properties.

The construction of filters we use is inspired by the ball lattice LSH of \cite{AI06}. In \cite{AI06}, an LSH family is constructed by covering a low-dimensional space $\RR^\dimC$ with lattices of balls and hashing each point to the first ball it is covered by. We apply this idea in the context of LSF, constructing an LSF family with Las Vegas properties in which each point is mapped to the set of balls containing it (i.e., the filters in the LSF family are the balls in the lattices of balls).

We combine this filter family with the general approach of \cite{Ahle17} for Las Vegas data structures for approximate near neighbors to obtain our data structure for Euclidean space. In \cite{Ahle17}, Ahle first constructs an LSF family with Las Vegas properties for a low-dimensional space $\{0, 1\}^\dimC$, where $\dimC = \Theta(\log\N)$. Then, a ``tensoring'' construction (as in \cite{Christiani17}) is used to combine LSF families for $\{0, 1\}^\dimC$ into an efficient LSF family for $\{0, 1\}^\dimB$, where $\dimB = \Theta(\log^{1+\beta}\N)$. However, directly tensoring with randomly sampled filter families is not sufficient for a Las Vegas data structure, and thus a key innovation of \cite{Ahle17} was to apply the splitters of \cite{NSS95, AMS06} as a derandomization tool. The splitters create a limited set of partitions of Hamming space, each of which corresponds to a way to tensor together lower-dimensional filter families such that the union of these tensored families is a filter family with Las Vegas properties.

We show that this approach extends from Hamming space to Euclidean space with some modifications. In particular, we adapt splitters to the geometric setting of $\RR^\dimB$ by constructing a collection $\Proj$ of orthogonal decompositions (see \Cref{subsec:orthdec}) of $\RR^\dimB$ such that each vector $x\in\RR^\dimB$ is ``split'' into components of almost equal length by at least one of the orthogonal decompositions in $\Proj$. The collection $\Proj$ is constructed with a new variation of CountSketch \cite{CCF04} using $2$-wise independent permutations \cite{AL13}. This geometric analog of splitters lets us obtain a similar result as \cite{Ahle17} for tensoring together filter families that cover the lower-dimensional space $\RR^\dimC$.

The final technique of dimensionality reduction is needed to go from efficient LSF families for $\RR^\dimB$ to a Las Vegas locality-sensitive data structure for $\RR^\dimA$, where the dimension $\dimA$ is arbitrary. The Johnson-Lindenstrauss lemma \cite{JL84} states that the dimension of a set of $\N$ points can be reduced to dimension $O(\eps^{-2}\log\N)$ while preserving pairwise distances up to a multiplicative factor of $1\pm\eps$. This idea is useful for high-dimensional approximate near neighbors and has a long history of being applied in various forms to such data structures (e.g., \cite{Indyk00, DIIM04, AI06, AC09, Ahle17, SW17}). However, as these dimensionality reduction maps are typically randomized in a Monte Carlo sense, \cite{Ahle17, SW17} require the dimensionality reduction to have an additional ``one-sided'' property to preserve the Las Vegas guarantee. Our construction relies on the same one-sided property and also improves on the runtimes of \cite{Ahle17, SW17} for this dimensionality reduction stage with a more careful two-stage sequence of reductions and an application of FastJL \cite{AC09}.

\section{Preliminaries}\label{sec:preliminaries}

\subsection{$\C$-Approximate Near Neighbors}

If $X = \RR^\dimA$ and $D$ is an $\ell_p$ norm, one can assume without loss of generality that $r = 1$, in which case we also refer to $(r, \C r)$-approximate near neighbors as \emph{$\C$-approximate near neighbors}. We will use this convention throughout the rest of this paper.

\subsection{Locality-Sensitive Filters}

The basic object of the LSF framework is a \emph{filter family}, which we define as follows:

\begin{defn}\label{defn:filterfamily}
  For a metric space $(X, D)$, a \emph{filter family} $\cF$ is a collection of subsets of $X$. An element $F\in\cF$ is known as a \emph{filter}. For a point $x\in X$, define $\cF(x)\coloneqq\{ F\in\cF : x\in F\}$ to be the set of all filters containing $x$. Note that a filter family can also be characterized by the values $\cF(x)$ for all $x\in X$, i.e., as a map from $X$ to the power set of filters.

\end{defn}

In \cite{Christiani17}, Christiani also defines $(r, \C r, p_1, p_2, p_q, p_u)$-sensitive filter families as distributions over filters with certain locality-sensitive properties. Although we will not be using this definition, since we require families that have a Las Vegas guarantee, we note that the filter families we consider in \Cref{sec:C,,sec:B} can be converted into a $(r, \C r, p_1, p_2, p_q, p_u)$-sensitive filter family by defining the distribution to be sampling a filter uniformly at random from the family. This reduction extends the lower bound proven in \cite{Christiani17} to the filter families we consider, meaning our later constructions of filter families are optimal. See \Cref{appendix:A} for more details about the lower bound.

\subsection{Orthogonal Decompositions}\label{subsec:orthdec}

Also useful to us are orthogonal decompositions of a vector space $\RR^d$, i.e., families of projections that express $\RR^d$ as a direct sum of $d'$-dimensional subspaces. Formally, we define these families as follows:

\begin{defn}
  A family $\{P_i\}_{i=1}^k$ of $d'\times d$ matrices is an \emph{orthogonal decomposition} of $\RR^d$ if $kd' = d$ and the set of row vectors of $P_1,\ldots,P_k$ forms an orthonormal basis of $\RR^d$.

\end{defn}

\section{Overview}\label{sec:overview}

As discussed in \Cref{subsec:techniques}, our construction can be roughly broken down into three stages. The first stage takes place in $\RR^\dimC$, where $\dimC = \Theta(\log^\expC\N)$ is a power of two and $0 < \expC < 1$; the second stage takes place in $\RR^\dimB$, where $\dimB = \Theta(\log^{1+\expB}\N)$ is a power of two and $0 < \expB < \expC$; and the third stage takes place in the original dimension $\RR^\dimA$. Although any $\expC$ and $\expB$ satisfying these constraints gives a data structure with exponent $\rho = 1 / \C^2 + o(1)$, the specific choice of $\expC$ and $\expB$ affects the $o(1)$ term in $\rho$. Setting $\expC = 4 / 5$ and $\expB = 2 / 5$ gives the fastest convergence of $\rho$ to $1 / \C^2$. (This choice yields a $o(1)$ term that diminishes as $\tilde O(\log^{-1 / 5}\N)$.)

The stages of our construction proceed as follows:
\begin{enumerate}
  \item
    \hyperref[sec:C]{\emph{Construct an LSF family that solves $\C$-approximate near neighbors in $\RR^\dimC$.}} We accomplish this by translating the ball lattice hashing approach of \cite{AI06} to the LSF framework. Our construction starts by sampling a collection of ball lattices to cover $\RR^\dimC$, with the individual balls being the filters of a filter family $\cF$. We then show that this sampled filter family has the desired ``Las Vegas'' locality-sensitive properties with probability at least $1 / 2$. We also describe an algorithm to verify these properties given a sampled $\cF$. It suffices for the sampling to succeed with probability $1/2$, since if the verification fails, we can restart---the constant probability of success implies we expect to restart at most $O(1)$ times.

    More concretely, the LSF family we construct has the following properties: a filter family $\cF$ can be constructed efficiently (in time $O(\poly(\dimC^\dimC)) = \N^{o(1)}$); for any $x\in\RR^\dimC$, $\cF(x)$ can be computed efficiently (in time $O(\poly(\dimC^\dimC))$); for all points $x, y\in\RR^\dimC$ such that $\dist{x}{y}\le 1$, $\cF(x)\cap\cF(y)\neq\emptyset$; and for all points $x, y\in\RR^\dimC$ such that $\dist{x}{y}\ge t$, $\EE_{\cF}[\abs{\cF(x)\cap\cF(y)}]\le\exp(-\Omega(t^2))$, with the constant implied by the big-$\Omega$ to be specified later.

  \item
    \hyperref[sec:B]{\emph{Construct an efficient LSF family that solves $\C$-approximate near neighbors in $\RR^\dimB$.}} We can construct a filter family for $\RR^\dimB$ by taking $\dimB / \dimC$ filter families for $\RR^\dimC$ along with an orthogonal decomposition of $\RR^\dimB$ into subspaces of dimension $\dimC$ and applying the tensor operation. In the tensored filter family, the set of filters containing a point $x\in\RR^\dimB$ is the direct product of the sets of filters containing each component of $x$ with respect to the orthogonal decomposition. We repeat this for a collection $\Proj$ of orthogonal decompositions of $\RR^\dimB$ into $\RR^\dimC$ and have our final filter family be the union of these tensored filter families.

    The collection $\Proj$ of orthogonal decompositions we consider has the property that every vector in $\RR^\dimB$ is ``split'' into components that are almost equal in length by some orthogonal decomposition in $\Proj$. To construct $\Proj$, we first describe a modification of CountSketch using $2$-wise independent permutations where each element of the modified CountSketch family is an orthogonal projection. We then construct the elements of $\Proj$ in a ``divide-and-conquer'' manner by composing CountSketch projections that each halve the dimension of the space. Important to us is the fact that $\Proj$ has size $\poly(\dimB^{\dimB / \dimC}) = \N^{o(1)}$.

    The splitting property of $\Proj$ implies that if $x, y\in\RR^\dimB$ are such that $\dist{x}{y}\le 1$, then there exists an orthogonal decomposition $\cP\in\Proj$ such that all components of $x - y$ under $\cP$ have length at most $(1 + \epsB)\sqrt{\dimC / \dimB} $ for $\epsB = \Theta(\log^{-\expC/2}\N)$. Thus, if a pair of points $x,y\in\RR^\dimB$ are ``close,'' then there exists an orthogonal decomposition $\cP\in\Proj$ such $x$ and $y$ are ``close'' in all subspaces given by $\cP$. This is useful because it lets the Las Vegas properties for the filter families in $\RR^\dimC$ transfer to the filter family in $\RR^\dimB$. Furthermore, an analogous property bounding the expected number of shared filters holds for ``distant'' points in $\RR^\dimB$.

    The result is a filter family for $\RR^\dimB$ with Las Vegas properties, such that the set of filters containing $x\in\RR^\dimB$ can be efficiently computed (because of tensoring). We can then use this filter family to obtain a data structure solving Euclidean $\C$-approximate near neighbors in $\RR^\dimB$ with exponent $\rho = 1 / \C^2 + o(1)$.

  \item
    \hyperref[sec:A]{\emph{Reduce approximate near neighbors in $\RR^\dimA$ to $\dimA / \dimB$ instances of approximate near neighbors in $\RR^\dimB$.}} We reduce the original $\C$-approximate near neighbors problem in $\RR^\dimA$ to $\dimA/\dimB$ instances of $\C'$-approximate near neighbors in $\RR^\dimB$, where $\C' = (1-\epsA)\C$ for $\epsA = \Theta(\log^{-\expB / 2}\N)$. This is done with a two-stage dimensionality reduction process that again uses orthogonal decompositions. We construct distributions over orthogonal decompositions with distributions of Johnson-Lindenstrauss maps in which all elements are projections.

    Orthogonal decompositions of $\RR^d$ into subspaces of dimension $d'$ have the property that if $x, y\in\RR^d$ are such that $\dist{x}{y}\le 1$, then at least one component of $x - y$ under the decomposition will have length at most $\sqrt{d' / d}$. Thus, if we treat each subspace as a separate instance of approximate near neighbors, then all $x$ and $y$ that are ``close'' in $\RR^d$ will also be close in at least one of the subspaces in the orthogonal decomposition, preserving the Las Vegas guarantee.

    A caveat of applying an orthogonal decomposition is that two ``distant'' points could be close in some subspaces, such that solving approximate near neighbors in those subspaces produces false positives for the original problem. We must filter out these false positives to maintain the Las Vegas guarantee. To reduce the time spent checking for false positives, we break the dimensionality reduction from $\RR^\dimA$ to $\RR^\dimB$ into two steps: The first step uses FastJL to map $\RR^{\dimA}$ into $\RR^{\dimB'}$ for $\dimB' = O(\poly(\log(\dimA\N)))$, and the second step uses a random projection to map $\RR^{\dimB'}$ into $\RR^\dimB$. This process has a runtime overhead of $O(d\poly(\log(\dimA\N)))$ per point. Hence query time, space usage, and preprocessing time are dominated by those of solving approximate near neighbors in the subspaces.
\end{enumerate}
In the following three sections, we describe in detail the constructions for each stage.

\section{Ball Lattice Filters for $\RR^\dimC$}\label{sec:C}

In this section, the variables $\dimC$, $\frad$, and $\fdel$ are each functions of the number of points $\N$, such that $\dimC = \Theta(\log^\expC\N)$, $\frad = \Theta(\log^{\expB / 2}\N)$ and $\fdel = \Theta(1 / \dimC)$ where $\expC$ and $\expB$ are as defined in \Cref{sec:overview}. (One can also consider the explicit parameter setting $\expC = 4 / 5$ and $\expB = 2 / 5$.)

We start with some notation and a lemma relating to balls in $\RR^\dimC$. For $x\in\RR^\dimC$, we use $\Ball(x, r)$ to denote the ball of radius $r$ centered at $x$. Furthermore, let $V_\dimC$ denote the volume of a unit $\dimC$-ball, let $C_\dimC(u)$ denote the volume of the cap at distance $u$ from the center of a unit $\dimC$-ball\footnote{Alternatively, $C_\dimC(u)$ is half the volume of the intersection of two unit $\dimC$-balls whose centers are distance $2u$ apart.}, and define $I_\dimC(u)\coloneqq C_\dimC(u) / V_\dimC$ to be the relative cap volume at distance $u$ for a $\dimC$-ball.

\begin{lem}[\cite{AI06}]\label{lem:balls}
  For all $\dimC\ge 2$ and $0\le u\le 1$,
  \[ \frac{1}{\sqrt \dimC}\bigp{1 - u^2}^{\frac\dimC 2}\lesssim I_\dimC(u)\le \bigp{1 - u^2}^{\frac\dimC 2}. \]
\end{lem}

The following proposition is the main result of this section; the distribution $\Filt$ over filter families with Las Vegas properties that we construct will be used in the tensoring step in \Cref{sec:B}.

\begin{prop}\label{prop:filt}
  There is a distribution $\Filt$ over filter families for $\RR^\dimC$ with the following properties:
  \begin{enumerate}
    \item\label{filt:sample}
      A filter family $\cF$ can be sampled from $\Filt$ in expected time $O(\poly(\dimC^\dimC))$.
    \item\label{filt:decode}
      For all $x\in\RR^\dimC$, $\cF(x)$ can be computed in expected time $O(\poly(\dimC^\dimC))$.
    \item\label{filt:all_pairs}
      For all $x, y\in\RR^\dimC$ such that $\dist{x}{y}\le 1$, $\cF(x)\cap\cF(y)\neq\emptyset$.
    \item\label{filt:local}
      Let $t\ge 0$ be a fixed constant. For all $x,y\in\RR^\dimC$ such that $\dist{x}{y}\ge t$,
      \[ \EE_{\cF\sim\Filt}(\abs{\cF(x)\cap\cF(y)}) \le O\bigp{\poly(\dimC)\exp\bigp{-\frac{\dimC}{2}\frac{t^2-1 - o(1)}{4\frad^2}}}. \]
      In particular, for $t = 0$, $\EE_{\cF\sim\Filt}(\abs{\cF(x)}) = O(\poly(\dimC)\exp(\frac{\dimC}{2}\frac{1 + o(1)}{4\frad^2}))$.
  \end{enumerate}
\end{prop}


Each filter in the filter families that we sample will be a ball of radius $\frad$. Like \cite{AI06}, we consider infinite lattices of such balls in $\RR^\dimC$, in particular translations of the lattice of balls of radius $\frad$ with centers at $3\frad\cdot\ZZ^\dimC$. Each such lattice of balls thus corresponds to an offset $v\in [0, 3\frad]^\dimC$.

Let $\Filt_0$ be the distribution over filter families such that a filter family $\cF$ is sampled by independently generating $\fnum$ offsets $v_1,\ldots,v_\fnum\in [0, 3\frad]^\dimC$ and defining
\[ \cF\coloneqq\bigcup_{i=1}^\fnum \bigc{\Ball(u, \frad) : u\in v_i + 3\frad\cdot\ZZ^\dimC}. \]
The set of filters a point belongs to is exactly the set of balls that contain it. (This differs from the ball lattice LSH of \cite{AI06}, in which a point gets hashed to the ball associated with the offset $v_i$ for the smallest $i$.)

A filter family sampled from $\Filt_0$ does not necessarily satisfy property \ref{filt:all_pairs}. Thus, to ensure property \ref{filt:all_pairs} holds, we obtain $\Filt$ from $\Filt_0$ by adding a verification step that drops all filter families drawn from $\Filt_0$ that do not satisfy property \ref{filt:all_pairs}. Checking that property \ref{filt:all_pairs} holds exactly is difficult, since there are infinitely many pairs $x, y\in\RR^{\dimC}$ such that $\dist{x}{y}\le 1$. We instead verify a stronger condition and choose $N$ so that a filter family sampled from $\Filt_0$ satisfies this stronger condition with probability at least $1 / 2$. Then a filter family can be sampled from $\Filt$ with $O(1)$ samples from $\Filt_0$ in expectation.

The stronger condition we check is the following: Consider the lattice $L\coloneqq\fdel\cdot\ZZ^\dimC$ and a smaller radius $\frad'\coloneqq\frad - \frac 12\fdel\sqrt\dimC$. We check that for all $x, y\in L$ such that $\dist{x}{y}\le 1 + \fdel\sqrt\dimC$, there exists some filter $F\in\cF$ with center $u$ such that $x, y\in\Ball(u, \frad')$. This condition suffices because for every pair $x', y'\in\RR^\dimC$ such that $\dist{x'}{y'}\le 1$, there exist $x, y\in L$ such that $\dist{x'}{x}, \dist{y'}{y}\le\frac 12\fdel\sqrt\dimC$. In particular, $\dist{x}{y}\le 1 + \fdel\sqrt\dimC$, and for all $u\in\RR^\dimC$, $x, y\in\Ball(u, \frad')$ implies $x', y'\in\Ball(u, \frad)$. Note that this condition holds if it holds on the subset $L\cap [0, 6\frad]^{\dimC}$ by the lattice structure of the filters. Therefore, it is enough to check $O((6\frad / \fdel)^{2\dimC}) = O(\poly(\dimC^\dimC))$ pairs of points in $L$, which can be done in $O(N\poly(\dimC^\dimC))$ time.

The next lemma lower bounds the probability of success for a fixed ``close'' pair $x, y\in L$:

\begin{lem}\label{lem:sample}
  Let $x, y\in L$ be such that $\dist{x}{y}\le 1 + \fdel\sqrt\dimC$. Suppose an offset $v\in [0, 3\frad]^\dimC$ is sampled uniformly at random. The probability there exists a point $u\in v + 3\frad\cdot\ZZ^\dimC$ such that $x, y\in\Ball(u, \frad')$ is at least $\Omega(\poly(\dimC^{-1})V_\dimC 3^{-\dimC}\exp(-\frac{\dimC}{2}\frac{1+o(1)}{4\frad^2}))$.
\end{lem}

\begin{proof}
  Let $E_1$ be the event that there exists a $u\in v + 3\frad\cdot\ZZ^\dimC$ such that $x, y\in \Ball(u, \frad')$ and $E_2$ be the event that there exists a $u\in v + 3\frad\cdot\ZZ^\dimC$ such that $x\in\Ball(u, \frad')$. Because $E_1$ implies $E_2$, we can condition to obtain
  \[ \PP(E_1) = \PP(E_2)\cdot\PP(E_1 | E_2)\ge\bigp{\frac{\frad'}{3\frad}}^\dimC V_\dimC\cdot 2I_\dimC\bigp{\frac{1 + \fdel\sqrt\dimC}{2\frad'}}. \]
  For the first term, assuming sufficiently large $\dimC$ and $\frad$:
  \begin{align*}
    \PP(E_2)
    &= V_\dimC 3^{-\dimC}\bigp{1 - \frac{\fdel\sqrt\dimC}{2\frad}}^\dimC \\
    &= V_\dimC 3^{-\dimC}\exp\bigp{\dimC\log\bigp{1 - \frac{\fdel\sqrt\dimC}{2\frad}}} \\
    &\ge V_\dimC 3^{-\dimC}\exp\bigp{-\frac{\dimC}{2}\frac{1}{4\frad^2}\cdot 8\frad\fdel\sqrt\dimC} \\
    &= V_\dimC 3^{-\dimC}\exp\bigp{-\frac{\dimC}{2}\frac{1}{4\frad^2}\cdot o(1)}.
  \end{align*}
  We now bound the second term. Since $I_\dimC(u)$ is decreasing in $u$, we first upper bound its argument (for large enough $\dimC$ and $\frad$) as $\frac{1}{2\frad'}(1 + \fdel\sqrt\dimC)\le\frac{1}{2\frad}(1 + 3\fdel\sqrt\dimC)$. Let $\xi = \frac{1}{2\frad}(1 + 3\fdel\sqrt\dimC)$. By \Cref{lem:balls},
  \begin{align*}
    \PP(E_1 | E_2)
    &\ge I_b(\xi) \\
    &\gtrsim\frac{1}{\sqrt\dimC} \bigp{1 - \xi^2}^{\frac{\dimC}{2}} \\
    &=\frac{1}{\sqrt\dimC}\exp\bigp{\frac\dimC 2\log\bigp{1 - \xi^2}} \\
    &\ge\frac{1}{\sqrt\dimC}\exp\bigp{-\frac\dimC 2\xi^2(1 + \xi^2)} \\
    &\ge\frac{1}{\sqrt\dimC}\exp\bigp{-\frac\dimC 2\frac{1}{4\frad^2}\bigp{1 + 3\fdel\sqrt\dimC}^2\bigp{1 + \frac{1}{\frad^2}}} \\
    &=\poly(\dimC^{-1})\exp\bigp{-\frac\dimC 2\frac{1}{4\frad^2} (1 + o(1))}.
  \end{align*}
  The desired bound now follows from combining our bounds on $\PP(E_1)$ and $\PP(E_1 | E_2)$.
\end{proof}

Let $p = \Omega(\poly(\dimC^{-1})V_\dimC 3^{-\dimC}\exp(-\frac{\dimC}{2}\frac{1 + o(1)}{4\frad^2}))$ be the least probability with which two ``close'' points share a filter in a random ball lattice as in \Cref{lem:sample}. Setting $\fnum = \Theta(p^{-1}\dimC\log(\frad / \fdel)) = O(\poly(\dimC^\dimC))$ allows all checks for $x, y\in L\cap [0, 6\frad]$ to succeed with probability at least $1 / 2$ by the union bound. Thus it is possible to check that a filter family sampled from $\Filt_0$ satisfies the stronger condition in $O(\poly(\dimC^\dimC))$, and property \ref{filt:sample} follows. To get property \ref{filt:decode}, note that there are only $\fnum = O(\poly(\dimC^\dimC))$ offsets, and thus computing the set of filters containing a given point can be done by iterating over all $\fnum$ offsets and finding for each offset the filter (if any) that contains the point.

It remains to verify property \ref{filt:local}. Let $x, y\in\RR^\dimC$ be such that $\dist{x}{y}\ge t\ge 0$. The construction succeeds with probability at least $1/2$. Therefore,
\[ \EE_{\cF\sim\Filt}(\abs{\cF(x)\cap\cF(y)}) = \EE_{\cF\sim\Filt_0}(\abs{\cF(x)\cap\cF(y)} | \,\text{construction succeeds})\le 2\EE_{\cF\sim\Filt_0}(\abs{\cF(x)\cap\cF(y)}). \]
We prove a lemma that lets us bound the expected value on the right-hand side:

\begin{lem}\label{filt:upper}
  Let $x, y\in\RR^\dimC$ be such that $\dist{x}{y}\ge t\ge 0$. Suppose an offset $v\in [0, 3\frad]^\dimC$ is sampled uniformly at random. The probability there exists a point $u\in v + 3\frad\cdot\ZZ^\dimC$ such that $x, y\in\Ball(u, \frad)$ is at most $V_\dimC 3^{-\dimC}\exp(-\frac\dimC 2\frac{t^2}{4\frad^2})$.
\end{lem}

\begin{proof}
  Let $E_1$ be the event that there exists a $u\in v + 3\frad\cdot\ZZ^\dimC$ such that $x, y\in \Ball(u, \frad)$ and $E_2$ be the event that there exists a $u\in v + 3\frad\cdot\ZZ^\dimC$ such that $x\in\Ball(u, \frad)$. By \Cref{lem:balls},
  \[ \PP(E_1) = \PP(E_2)\cdot\PP(E_1 | E_2)\le V_\dimC 3^{-\dimC}\cdot I_b\bigp{\frac{t}{2\frad}}\le V_\dimC 3^{-\dimC}\bigp{1 - \frac{t^2}{4\frad^2}}^{\frac\dimC 2}\le V_\dimC 3^{-\dimC}\exp\bigp{-\frac\dimC 2\frac{t^2}{4\frad^2}}. \]
\end{proof}

Applying \Cref{filt:upper}, we get that
\[ \EE_{\cF\sim\Filt_0}(\abs{\cF(x)\cap\cF(y)})\le\fnum\cdot V_\dimC 3^{-\dimC}\exp\bigp{-\frac\dimC 2\frac{t^2}{4\frad^2}} = O\bigp{\poly(\dimC)\exp\bigp{-\frac\dimC 2\frac{t^2 - 1 - o(1)}{4\frad^2}}}, \]
which gives us property \ref{filt:local} and completes our proof of \Cref{prop:filt}.

\section{Tensoring Up}\label{sec:B}

In this section, we use a ``tensoring'' operation and a Euclidean analog of splitters \cite{NSS95, AMS06} to construct a distribution over filter families for $\RR^\dimB$. This distribution will give us an efficient data structure for Euclidean $\C$-approximate near neighbors in $\RR^\dimB$.  The subsections consist of constructing a collection $\Proj$ of orthogonal decompositions with a ``splitting'' property and showing how to use $\Proj$, tensoring, and \Cref{prop:filt} to get the desired distribution over filter families for $\RR^\dimB$.

\subsection{A Collection of ``Splitting'' Orthogonal Decompositions}\label{subsec:splitting}

In this subsection, we assume without loss of generality that both the initial dimension $d$ and the dimension $d'$ ($d' < d$) of the space that we decompose $\RR^d$ into are powers of two. We can ensure this by padding zeroes to either space while increasing the dimension by at most a constant factor.

We describe how to construct a collection $\Proj$ of orthogonal decompositions of $\RR^d$ into $\RR^{d'}$ such that for any vector $x\in\RR^d$, there is an orthogonal decomposition in $\Proj$ that ``splits'' $x$ into components in $\RR^{d'}$ that are almost equal in length. Using a modified CountSketch \cite{CCF04}, we get this splitting property while having only $\poly(d^{d / {d'}})$ orthogonal decompositions in $\Proj$.

\begin{prop}\label{prop:splitting}
  For all $0 < \eps < 1 / 2$ and $d' = \Omega(1 / \eps^2)$, there exists a collection $\Proj$ of orthogonal decompositions of $\RR^{d}$ into $\RR^{d'}$ such that $\abs{\Proj} = O(\poly(d^{d / d'}))$ and for all $x\in\RR^d$ of unit norm, there exists a decomposition $\cP\in\Proj$ such that
  \[ \abs{\sqrt{\frac{d}{d'}}\norm{Px}_2 - 1} < \eps \]
  for all $P\in\cP$.
\end{prop}

To construct such a collection $\Proj$, we introduce a variation of CountSketch using $2$-wise independent permutations \cite{AL13}. The traditional CountSketch implementation \cite{CCF04, CJN18} gives a family of linear maps from $\RR^d$ to $\RR^{d'}$, where $d' = \Omega(1 / (\eps^2\delta))$, such that for all $x\in\RR^d$, the distortion is within $1\pm\eps$ with probability at least $1 - \delta$. We prove a similar result for a modified CountSketch where each element of the family is a projection. CountSketch can be derandomized with $k$-wise independent hash families; we do the same for our variant with $2$-wise independent permutations.

Let $\eps = 4 / \sqrt{d'}$ and let $A_0$ be the $d'\times d$ matrix
\[
  d'\left\{\begin{matrix}
      \vphantom{} \\ \vphantom{} \\ \vphantom{} \\ \vphantom{} \\
      \vphantom{} \\ \vphantom{} \\ \vphantom{} \\
  \end{matrix}\right.
  \underbrace{\begin{bmatrix}
    \smash[b]{\block{d / d'}} & \\
                              & \smash[b]{\block{d / d'}} \\
                              &                           & \ddots \\
                              &                           &        & \block{d / d'}
  \end{bmatrix}}_{d}.
\]
Then $A_0$ is a matrix all of whose non-zero entries are $\sqrt{d' / d}$, has exactly one non-zero entry per column, and has exactly $d / d'$ non-zero entries per row. Let $H$ be a family of $2$-wise independent permutations of $[d]$ and $\Sigma$ be a family of $4$-wise independent hash functions $[d]\to\{\pm 1\}$. Let $A_{h,\sigma}$ for $h\in H$ and $\sigma\in\Sigma$ be the matrix obtained by permuting the columns of $A_0$ by $h$ and then multiplying the $i$-th column by $\sigma(i)$ for each $i\in [d]$. We define our modified CountSketch family to be $\cA\coloneqq\{ A_{h,\sigma} : h\in H, \sigma\in\Sigma \}$. To analyze this family, we use the second-moment method.

\begin{lem}\label{lem:CountSketch}
  Suppose $A$ is sampled uniformly at random from $\cA$. For all $x\in\RR^d$ of unit norm,
  \[ \PP\bigp{\abs{\sqrt{\frac{d}{d'}}\norm{Ax}_2 - 1} > \eps} < \frac 12. \]
\end{lem}

\begin{proof}
  Let $\eta_{r,i}$ be an indicator for whether $A_{r,i}$ is non-zero, and let $\sigma_i$ be the sign of the non-zero entry of $A_{\cdot,i}$. Then for $A$ sampled uniformly at random, the $\eta_{\cdot,j}$ vectors are drawn from a $2$-wise independent permutation of $d$ elements and the $\sigma_i$ values are $4$-wise independent, with the sets of values $\{\eta_{\cdot,\cdot}\}$ and $\{\sigma_\cdot\}$ themselves being independent of each other.

  For the event to not occur, it suffices that $\frac{d}{d'}\norm{Ax}_2^2 - 1\in [-2\eps + \eps^2, 2\eps]$. Therefore, it suffices to analyze the random variable $Z\coloneqq\frac{d}{d'}\norm{Ax}_2^2 - 1$. Rewrite
  \[ Z = \sum_{r=1}^{d'}\sum_{i\neq j} \eta_{r,i}\eta_{r,j}\sigma_{i}\sigma_{j} x_ix_j. \]
  By Chebyshev's inequality,
  \begin{align*}
    \PP(\abs{Z} > \eps)
    &\le\frac{1}{\eps^2}\EE(Z^2) \\
    &= \frac{1}{\eps^2}\EE\bigp{\sum_{r=1}^{d'}\bigp{\sum_{i\neq j} \eta_{r,i}\eta_{r,j}\sigma_i\sigma_jx_ix_j}^2} \\
    &\qquad\qquad + \frac{1}{\eps^2}\EE\bigp{\sum_{r\neq s}\bigp{\bigp{\sum_{i\neq j} \eta_{r,i}\eta_{r,j}\sigma_i\sigma_jx_ix_j}\bigp{\sum_{i\neq j} \eta_{s,i}\eta_{s,j}\sigma_i\sigma_jx_ix_j}}}.
  \end{align*}
  Due to the $4$-wise independence of the $\sigma_i$ and the fact that $\eta_{r,i}\eta_{s,i} = 0$ when $r\neq s$, we have that the expectations of all terms in the second summation are zero. By the same reasoning, the expectation of the first term is equal to
  \[ 2d'\sum_{i\neq j} \EE(\eta_{r,i}\eta_{r,j}) x_i^2x_j^2 = 2d'\sum_{i\neq j}\frac{d / d'}{d}\frac{d / d' - 1}{d} x_i^2x_j^2\le 2\bigp{\frac{1}{d'} - \frac{1}{d}}\bigp{\sum_{i} x_i^2}^2\le\frac{2}{d'}. \]
  Thus $\PP(\abs Z > \eps)\le 2 / (\eps^2 d') < 1 / 2$.
\end{proof}

\begin{figure}\label{fig:decomp}
  \centering
  \begin{tikzcd}[column sep=-1.6em]
    & &  &  &  &  &  & \boxed{\quad\quad\quad\quad x\quad\quad\quad\quad} \arrow[llld, "A_0"'] \arrow[rrrd, "A_1"] &  &  &  &  &  &  \\
    & &  &  & \boxed{\quad\quad x_0\quad\quad} \arrow[lld, "A_{00}"'] \arrow[rd, "A_{01}"] &  &  &  &  &  & \boxed{\quad\quad x_1\quad\quad} \arrow[ld, "A_{10}"'] \arrow[rrd, "A_{11}"] &  &  &  \\
    & & \boxed{\quad x_{00}\quad} \arrow[rd, "A_{001}"] \arrow[ld, "A_{000}"'] &  &  & \boxed{\quad x_{01}\quad} \arrow[rd, "A_{011}"] \arrow[ld, "A_{010}"'] &  &  &  & \boxed{\quad x_{10}\quad} \arrow[rd, "A_{101}"] \arrow[ld, "A_{100}"'] &  &  & \boxed{\quad x_{11}\quad} \arrow[ld, "A_{110}"'] \arrow[rd, "A_{111}"] &  \\
    & \quad\vdots\arrow[ld, "A_{0\cdots 00}"']\arrow[d, "A_{0\cdots 01}"] &  & \vdots & \quad\quad\vdots &  & \vdots &  & \vdots &  & \vdots\quad\quad & \vdots &  & \vdots\arrow[d, "A_{1\cdots 10}"']\arrow[rd, "A_{1\cdots 11}"]\quad \\
    \boxed{x_{0\cdots 00}}\quad & \boxed{x_{0\cdots 01}} &  &  & \cdots &  &  & \cdots &  &  & \cdots & & & \boxed{x_{1\cdots 10}} & \quad\boxed{x_{1\cdots 11}}
  \end{tikzcd}
  \caption{An orthogonal decomposition $\cP\in\Proj$ applied to some vector $x\in\RR^d$.}
\end{figure}

We now construct the collection $\Proj$ of orthogonal decompositions of $\RR^d$ into $\RR^{d'}$. Let $\ell\in\ZZ$ be such that $d = 2^\ell d'$. Let $d_j = 2^j d'$ and $\eps_j = 10 / \sqrt{d_\ell}$ for $j\in\{0,\ldots,\ell\}$. Let $\cA_1,\ldots,\cA_\ell$ be families of modified CountSketch projections as defined above, where $\cA_i$ consists of linear maps $\RR^{d_i}\to\RR^{d_{i-1}}$.

Each orthogonal decomposition in $\Proj$ can be constructed as follows: Start with some $A_0\in\cA_\ell$. Let $A_1\colon\RR^{d_\ell}\to\RR^{d_{\ell-1}}$ be a projection onto $\ker(A_0)$ (equivalently a projection onto the orthogonal complement of the row space of $A_0$). We then choose $A_{00}, A_{10}\in\cA_{\ell-1}$, and set $A_{01}$ and $A_{11}$ to be projections onto $\ker(A_{00})$ and $\ker(A_{10})$, respectively. Similarly, we choose $A_{s_1\cdots s_{i-1}0}\in\cA_{\ell-i+1}$ for each $(s_1,\ldots,s_{i-1})\in\{0,1\}^{i-1}$ for all $i\in [\ell]$ and set $A_{s_1\cdots s_{i-1}1}$ to be a projection onto $\ker(A_{s_1\cdots s_{i-1}0})$. For each $(s_1,\ldots,s_\ell)\in \{0,1\}^\ell$, define $P_{s_1\cdots s_\ell}\coloneqq A_{s_1\cdots s_\ell}A_{s_1\cdots s_{\ell-1}}\cdots A_{s_1s_2}A_{s_1}$. We take the set
\[ \cP\coloneqq\bigc{ P_{s_1\cdots s_\ell} : (s_1,\ldots,s_\ell)\in\{0,1\}^\ell} \]
to be an element of $\Proj$. (In particular, it is not difficult to see that $\cP$ is an orthogonal decomposition of $\RR^d$.) For a diagram of this process, see \Cref{fig:decomp}.

We now prove that the set $\Proj$ of all such orthogonal decompositions $\cP$ (i.e., over all choices of elements from the $\cA_i$) has the desired properties for \Cref{prop:splitting}. Suppose $x\in\RR^d$ is of unit norm. By \Cref{lem:CountSketch} there is some $A_0\in\cA_\ell$ that projects $x$ to $x_0\in\RR^{d_{\ell-1}}$ with distortion $1\pm\eps_\ell$. Then $A_1$ also projects $x$ to $x_1\in\RR^{d_{\ell-1}}$ with distortion $1\pm\eps_\ell$ by the Pythagorean theorem. Inducting on $i$ while applying \Cref{lem:CountSketch}, note that at each level there exists some $A_{s_1\cdots s_{i-1}0}\in\cA_i$ that projects $x_{s_1\cdots s_{i-1}}$ with distortion $1\pm\eps_{\ell - i + 1}$ to $x_{s_1\cdots s_{i-1}0}$. Summing the geometric series, observe that the total distortion of any $x_{s_1\cdots s_\ell}$ relative to $x$ in this construction is bounded by
\[ \exp\bigp{\log(1\pm\eps_\ell) + \cdots + \log(1\pm\eps_1)} = \exp(O(\pm\eps_\ell\pm\cdots\pm\eps_1)) = 1\pm O(\eps). \]
Therefore, the orthogonal decomposition $\cP\coloneqq\{P_{s_1\cdots s_\ell}x : (s_1,\ldots,s_\ell)\in\{0,1\}^\ell\}$ satisfies the distortion property required for \Cref{prop:splitting}.

Finally, we check that $\abs\Proj$ is small enough. Note that each element of $\cA_i$ can be characterized by a permutation $h$ on $d$ elements drawn from a $2$-wise independent permutation family and a function $\sigma\colon [d]\to\{\pm 1\}$ drawn from a $4$-wise independent hash family. Both $h$ and $\sigma$ can be characterized by $O(\log d)$ bits, so $\abs{\cA_i} = O(\poly(d))$ for all $i$. Each element of $\Proj$ is defined in terms of $2^\ell - 1 \le d / d'$ selections from $\cA_i$ families. Hence $\abs\Proj = O(\poly(d)^{d/d'}) = O(\poly(d^{d/d'}))$, and our construction of $\Proj$ for \Cref{prop:splitting} is complete.

\subsection{Efficient Filter Families for $\RR^\dimB$}

In this subsection, $\dimB$ and $\epsB$ are functions of $\N$ such that $\dimB = \Theta(\log^{1 + \expB}\N)$ and $\epsB = \Theta(\log^{-\expC / 2}\N)$, where $\expC$ and $\expB$ are as defined in \Cref{sec:overview}. (One can also consider the explicit parameter setting $\expC = 4 / 5$ and $\expB = 2 / 5$.) Furthermore, $\frad = \Theta(\log^{\expB/2}\N)$ is as defined in \Cref{sec:C}.

We construct a filter family with Las Vegas properties for $\RR^\dimB$ via a tensoring operation. The tensoring operation produces a filter family for $\RR^\dimB$ by combining $\dimB / \dimC$ filter families for $\RR^\dimC$ with respect to an orthogonal decomposition of $\RR^\dimB$. We take the union of several tensored filter families, one for each element of $\Proj$, to obtain the filter family for $\RR^\dimB$ that we want. The splitting property of $\Proj$ gives this filter family the desired Las Vegas properties.


We begin by defining what it means to ``tensor'' together filter families and then state the main result (\Cref{prop:filtB}) for this subsection.
\begin{defn}
  The \emph{tensoring} operation takes an orthogonal decomposition $\{P_i\}_{i=1}^{d / {d'}}$ of $\RR^d$ into $\RR^{d'}$ and filter families $\cF_1,\ldots,\cF_{d / {d'}}$ for $\RR^{d'}$ and returns a filter family $\cG$ whose filters are such that
  \[ \cG(x) = \cF_1\bigp{\sqrt{\frac{d}{{d'}}} P_1 x}\times\cdots\times\cF_{d / {d'}}\bigp{\sqrt{\frac{d}{{d'}}} P_{d / {d'}} x}. \]
  In particular, the filters of $\cG$ are implicitly characterized by defining $\cG(x)$ for each $x\in\RR^d$.
\end{defn}

\begin{prop}\label{prop:filtB}
  There is a distribution $\FiltB$ over filter families for $\RR^\dimB$ with the following properties:
  \begin{enumerate}
    \item\label{filtB:sample}
      A filter family $\cG$ can be sampled from $\FiltB$ in expected time $O(\poly(\dimB^{\dimB / \dimC}\dimC^\dimC))$.
    \item\label{filtB:decode}
      For all $x\in\RR^\dimB$, $\cG(x)$ can be computed in expected time $O(\poly(\dimB^{\dimB / \dimC}\dimC^\dimC) + \abs{\cG(x)})$. 
    \item\label{filtB:all_pairs}
      For all $x, y\in\RR^\dimB$ such that $\dist{x}{y}\le 1$, $\cG(x)\cap\cG(y)\neq\emptyset$.
    \item\label{filtB:local}
      Let $t\ge 0$ be a fixed constant. For all $x,y\in\RR^\dimB$ such that $\dist{x}{y}\ge t$,
      \[ \EE_{\cG\sim\FiltB}(\abs{\cG(x)\cap\cG(y)}) \le O\bigp{\poly(\dimB^{\dimB / \dimC})\exp\bigp{-\frac{\dimB}{2}\frac{t^2-1 - o(1)}{4\frad^2}}}. \]
      In particular, for $t = 0$, $\EE_{\cG\sim\FiltB}(\abs{\cG(x)}) = O(\poly(\dimB^{\dimB / \dimC})\exp(\frac{\dimB}{2}\frac{1 + o(1)}{4\frad^2}))$.
  \end{enumerate}
\end{prop}


Let $\Filt$ be as in \Cref{prop:filt} and $\Proj$ be as in \Cref{prop:splitting}, with the elements of $\Proj$ decomposing $\RR^{\dimB}$ into copies of $\RR^{\dimC}$. To sample a filter family from $\FiltB$, we first define $\cG_\cP$ as a filter family obtained by applying the tensoring operation to $\cP\in\Proj$ and $\dimB / \dimC$ filter families drawn from $\Filt$. Now, define $\cG_0\coloneqq\bigcup_{\cP\in\Proj}\cG_\cP$. Our sampled family $\cG$ is defined so that $\cG(x) = \cG_0(x / (1 + \epsB))$.

Such a $\cG$ can be sampled in time $\abs{\Proj}\frac{\dimB}{\dimC}O(\poly(\dimC^\dimC)) = O(\poly(\dimB^{\dimB / \dimC}\dimC^\dimC))$, which yields property \ref{filtB:sample}. To compute $\cG(x)$, we first compute $\cF(x)$ in time $O(\poly(\dimC^\dimC))$ for each $\cF$ sampled from $\Filt$. There are $\poly(\dimB^{\dimB / \dimC})$ such $\cF$, so the time for this step is $O(\poly(\dimB^{\dimB / \dimC}\dimC^\dimC))$. Finally, $\cG(x)$ can be generated from the sets $\cF(x)$ in $O(\abs{\cG(x)})$, giving us property \ref{filtB:decode}.

To see property \ref{filtB:local}, suppose $x, y\in\RR^\dimB$ are such that $\dist{x}{y}\ge t\ge 0$. Let $x' = x / (1 + \epsB)$ and $y' = y / (1 + \epsB)$. Then for each $\cG_\cP$, where $\cP\coloneqq\{P_i\}_{i=1}^{\dimB / \dimC}\in\Proj$ and $\cF_1,\ldots,\cF_{\dimB / \dimC}$ are sampled from $\Filt$,
\begin{align*}
  \EE(\abs{\cG_\cP(x')\cap\cG_\cP(y')})
  &= \prod_{i=1}^{\dimB / \dimC} \EE_{\cF_i\sim\Filt}\bigp{\abs{\cF_i\bigp{\sqrt{\frac{\dimB}{\dimC}} P_ix'}\cap\cF_i\bigp{\sqrt{\frac{\dimB}{\dimC}} P_iy'}}} \\
  &\le\prod_{i=1}^{\dimB / \dimC} O\bigp{\poly(\dimC)\exp\bigp{-\frac{\dimC}{2}\frac{\frac{\dimB}{\dimC}\dist{P_ix'}{P_iy'}^2 - 1 - o(1)}{4\frad^2}}} \\
  &= O\bigp{\poly(\dimC^{\dimB / \dimC})\exp\bigp{-\frac{\dimB}{2}\frac{t^2 - 1 - o(1)}{4\frad^2}}}.
\end{align*}
Since $\EE_{\cG\sim\FiltB}(\abs{\cG(x)\cap\cG(y)}) = \poly(\dimB^{\dimB / \dimC})\EE(\abs{\cG_{\cP}(x')\cap\cG_{\cP}(y')})$, the above gives us property \ref{filtB:local}.

Finally, to verify property \ref{filtB:all_pairs}, suppose $x, y\in\RR^\dimB$ are such that $\dist{x}{y}\le 1$. Let $x' = x / (1 + \epsB)$ and $y' = y / (1 + \epsB)$. By the splitting property of $\Proj$ (see \Cref{prop:splitting}), there exists some $\cP$ in $\Proj$ such that for all $P\in\cP$,
\[ \sqrt{\frac{\dimB}{\dimC}}\norm{Px' - Py'}_2\le 1. \]
Then $\cG(x)\cap\cG(y)\supseteq\cG_\cP(x')\cap\cG_\cP(y')\neq\emptyset$ by the tensor definition of $\cG_\cP$ and property \ref{filt:all_pairs} of $\Filt$.

\begin{cor}\label{cor:B}
  There exists a Las Vegas data structure for Euclidean $\C$-approximate near neighbors in $\RR^\dimB$ with expected query time $O(\N^{\rho})$, space usage $O(\N^{1 + \rho})$, and preprocessing time $O(\N^{1 + \rho})$ for $\rho = 1 / \C^2 + o(1)$.
\end{cor}

\begin{proof}
  We use the distribution $\FiltB$ constructed in \Cref{prop:filtB}. Recall our parameter settings for $\dimC$ and $\dimB$ in terms of $\expC$ and $\expB$, i.e., $\dimC = \Theta(\log^{\expC}\N)$ and $\dimB = \Theta(\log^{1 + \expB}\N)$. Notice that $\poly(\dimB^{\dimB / \dimC}\dimC^\dimC) = \N^{o(1)}$. Recall also that $\frad = \Theta(\log^{\expB / 2}\N)$. For $\frad$, we specifically set
  \[ \frad\coloneqq\sqrt{\frac{\dimB}{8\log\N}}\C. \]
  By \Cref{prop:filtB}, for $\cG\sim\FiltB$, the expected number of filters containing any $x\in\RR^{\dimB}$ is $O(\N^{1/\C^2 + o(1)})$. In preprocessing, we sample $\cG$ and compute $\cG(x)$ for each $x$ in $\cD$, which takes time $O(\N^{1 + 1/\C^2 + o(1)})$. Because we need $\abs{\cG(x)}$ space to store the filters containing $x$, space usage is also $O(\N^{1 + 1 / \C^2 + o(1)})$. Finally, to answer queries with this data structure, we can compute $\cG(q)$ for a query point $q\in\RR^\dimB$ in expected time $O(n^{1 / \C^2 + o(1)})$. The expected number of collisions between $q$ and distant points in $\cD$ is $n\cdot O(n^{-1 + 1 / \C^2 + o(1)}) = O(n^{1 / \C^2 + o(1)})$. Hence the expected runtime of a query is $O(n^{1 / \C^2 + o(1)})$.
\end{proof}

\section{Projecting Down}\label{sec:A}

We complete our construction of the data structure in this section by defining a two-stage sequence of dimensionality reductions that efficiently reduces the original problem into $\dimA / \dimB$ instances of approximate near neighbors in $\dimB$ dimensions. In this reduction, it is guaranteed that any pair of near neighbors are considered near neighbors in \emph{at least} one of the lower-dimensional problems.

Our dimensionality reduction improves on previous implementations of the ``one-sided'' dimensionality reduction idea \cite{Ahle17, SW17, Wygocki17} in that the overhead per point is now $\tilde O(\dimA)$ instead of having a linear dependence on $\dimA^2$. The data structure in \cite{Ahle17} implements this idea for Hamming space and has $O(1)$ false positives from each lower-dimensional problem. Each false positive requires $O(\dimA)$ time to filter out, resulting in a total cost of $O(\dimA^2 / \dimB)$. Our two-stage dimensionality reduction lets us spend only $\tilde O(\dimA)$ time handling false positives. The data structures of \cite{SW17} and \cite{Wygocki17} implement one-sided dimensionality reduction for Euclidean space, but require $\Omega(\dimA^2)$ time to compute the dimensionality reduction since they multiply by a random rotation in $\RR^\dimA$. We note that this can be improved to $\tilde O(\dimA)$ by applying a slightly modified FastJL.



After applying the dimensionality reduction, we use the data structure constructed for $\RR^\dimB$ from \Cref{sec:B} to get a Las Vegas data structure that solves $\C$-approximate near neighbors for all dimensions with an exponent of $\rho = 1 / \C^2 + o(1)$.

\subsection{Efficient Distributions of Orthogonal Decompositions}

In this subsection, we prove two ``one-sided'' dimensionality reduction results (\Cref{prop:fastJL} and \Cref{lem:JL}) using orthogonal decompositions. These results will be used to reduce the dimension of the problem from $\dimA$ to $\dimB$.

For this subsection, we assume without loss of generality that the initial dimension $d$ is a power of two. We can guarantee this by padding zeros, which increases $d$ by at most a constant factor.

\begin{prop}\label{prop:fastJL}
  For all $0 < \eps < 1 / 2$, $0 < \delta < 1 / 2$, and $d' = \Omega(\eps^{-2}\log(1 / (\eps\delta))\log(1 / \delta))$, there exists a distribution $\mathscr P_1$ over orthogonal decompositions of $\RR^{d}$ into $\RR^{d'}$ such that for all $x\in\RR^{d}$ of unit norm and $i\in [d / d']$, where $\cP\coloneqq\{P_i\}_{i=1}^{d / d'}$ is sampled from $\mathscr P_1$,
  \[ \PP_{\cP\sim\mathscr P_1}\bigp{\abs{\sqrt{\frac{d}{d'}}\norm{P_ix}_2 - 1} > \eps} < \delta. \]
  Furthermore, the set of values $\{P_ix\}_{i=1}^{d / d'}$ is computable in $O(d\log d)$ time.
\end{prop}

To construct this distribution $\mathscr P_1$, we use a modified FastJL \cite{AC09}. Let $H_d$ be the $d$-dimensional Hadamard matrix, and consider the distribution $\cA$ over $d'\times d$ matrices defined as follows: Sample a $d\times d$ diagonal matrix $D$ whose diagonal entries are independent Rademachers and a $d'\times d$ matrix $S$ whose rows are standard basis vectors sampled without replacement. We take $A\coloneqq SH_d D$ to be the sampled element of $\cA$. Note that $A$ is orthogonal. We analyze $\cA$ with the following lemma:

\begin{lem}\label{lem:fastJL}
  Suppose $A\in\cA$ is sampled uniformly at random. Then for all $x\in\RR^d$ of unit norm,
  \[ \PP\bigp{\abs{\sqrt{\frac{d}{d'}}\norm{Ax}_2 - 1} > \eps} < \delta. \]
\end{lem}

\begin{proof}
  Let $\cA'$ be the distribution $\cA$ with the modification that the rows of the coordinate sampling matrix $S$ are sampled \emph{with} replacement. Applying \cite[Theorem 4]{Hoeffding63} and \cite[Theorem 9]{CNW16}, we have for $p = \log(1 / \delta)$ that
  \[ \EE_{A\sim\cA}\bigp{\abs{\frac{d}{d'}\norm{Ax}_2^2 - 1}^p}\le\EE_{A\sim\cA'}\bigp{\abs{\frac{d}{d'}\norm{Ax}_2^2 - 1}^p} < \eps^p\delta, \]
  since $x\mapsto\abs{x}^p$ is a convex function. By Markov's inequality, the lemma follows.
\end{proof}

We now construct our distribution $\mathscr P_1$ over orthogonal decompositions by sampling $D\in\RR^{d\times d}$ and a random permutation matrix $P\in\RR^{d\times d}$. We define the orthogonal decomposition to be the rows of $PH_dD$ partitioned into $d / d'$ matrices of dimension $d'\times d$. Since $PH_dD$ is orthogonal, the resulting family is an orthogonal decomposition. By \Cref{lem:fastJL}, this distribution over orthogonal decompositions satisfies the condition of \Cref{prop:fastJL}.

\begin{lem}[\cite{SW17}]\label{lem:JL}
  For all $0 < \eps < 1 / 2$, $0 < \delta < 1 / 2$, and $d' = \Omega(\eps^{-2}\log(1 / \delta))$, there exists a distribution $\mathscr P_2$ over orthogonal decompositions of $\RR^d$ into $\RR^{d'}$ such that for all $x\in\RR^d$ of unit norm and $i\in [d / d']$, where $\cP = \{P_i\}_{i=1}^{d / d'}$ is sampled from $\mathscr P_2$,
  \[ \PP_{\cP\sim\mathscr P_2}\bigp{\abs{\sqrt{\frac{d}{d'}}\norm{P_i x}_2 - 1} > \eps} < \delta. \]
\end{lem}

\begin{proof}
  Sample a random rotation $R$ in $d$-dimensions and define the orthogonal decomposition to be the rows of $R$ split into $\frac{d}{d'}$ matrices of dimension $d'\times d$. By the Johnson-Lindenstrauss lemma, if $d' = \Omega(\eps^{-2}\log(1 / \delta))$, then the property holds with the desired probability.
\end{proof}

\subsection{Constructing the Data Structure for $\RR^\dimA$}

Let $\epsA$ be a function of $\N$ with $\epsA = \Theta(\log^{-\expB / 2}\N)$, where $\expB$ is as defined in \Cref{sec:overview}. (One can also consider the explicit parameter setting $\expB = 2 / 5$.)


Our algorithm for $\C$-approximate near neighbors works by a sequence of two reductions to lower dimensional $\C$-approximate near neighbors problems. In the first reduction, we reduce the dimension by applying an orthogonal decomposition drawn from $\mathscr P_1$ (as in \Cref{prop:fastJL}) for $\delta = 1 / (\N\dimA)$ and $\eps$ as above to obtain $\dimA / \dimB'$ problems in dimension $\dimB'\coloneqq\Theta(\eps^{-2}\log^2(\N\dimA))$. We then reduce the dimension further by applying an orthogonal decomposition drawn from $\mathscr P_2$ (as in \Cref{lem:JL}) with $\delta = 1 / \N$ and $\eps$ as above for each $\dimB'$-dimensional problem to obtain a set of $\dimB$-dimensional problems. Finally, we apply the algorithm of \Cref{cor:B}. We describe the steps in greater detail below.

\subsubsection{A General Reduction}

We give a general reduction from $\C$-approximate near neighbors in $\RR^d$ to $d / d'$ instances of $\C'$-approximate near neighbors in $\RR^{d'}$ for $\C'\coloneqq\C(1 - \epsA)$, assuming a distribution $\mathscr P$ over orthogonal decompositions of $\RR^d$ into $\RR^{d'}$ such that for $\cP\coloneqq\{P_i\}_{i=1}^{d / d'}$ sampled from $\mathscr P$,
\[ \PP_{\cP\sim\mathscr P}\bigp{\abs{\sqrt{\frac{d}{d'}}\norm{P_ix}_2 - 1} > \eps} < \delta. \]

We start by sampling an orthogonal decomposition $\cP\sim\mathscr P$. Then $\cP$ decomposes $\RR^d$ as a direct sum of $d / d'$ subspaces of dimension $d'$. For any two points $x, y\in\RR^d$ such that $\dist{x}{y}\le 1$, if we let $x_1,\ldots,x_{d / d'}$ and $y_1,\ldots, y_{d/d'}$ be the components of $x$ and $y$ under $\cP$, respectively, then there exists some $i$ such that $\sqrt{d / d'}\dist{x_i}{y_i}\le 1$. Therefore, if $x$ and $y$ are near neighbors in $\RR^d$, then they are also near neighbors in at least one of subspaces (after scaling by $\sqrt{d / d'}$). That is, all pairs of near neighbors in $\RR^d$ correspond to at least one pair of near neighbors in the subspaces.

Our only worry is that $\dist{x}{y} > \C$, but $\dist{x_i}{y_i}\le (1 - \eps)\C$, i.e., we have a \emph{false positive}. By the definition of $\mathscr P$, the expected number of false positives is bounded above by $\delta\cdot \N d / d'$ for a query $y\in\RR^\dimA$. However, we can filter out these false positives out by continuing on in the algorithm for the lower-dimensional problem as if a false positive $x$ did not share a filter with $y$. The cost of filtering out such a false positive is $O(d)$, since we need to check whether $\dist{x}{y} > c$ for each candidate near neighbor $x$ returned. Thus the total cost of false positives is $\delta\N d^2 / d'$.


\subsubsection{Proof of \Cref{thm:main}}

We now prove \Cref{thm:main} by describing a data structure in which the above reduction is applied twice. We first apply the reduction to get from dimension $\dimA$ to dimension $\dimB'$ using $\mathscr P_1$ and $\delta\coloneqq 1 / (\N\dimA)$. For this, we need time $\delta\N\dimA^2 / \dimB' = O(\dimA / \dimB')$ to handle false positives and time $O(\dimA\log\dimA)$ to compute the orthogonal decomposition of $q$. The space overhead of storing an element of $\mathscr P_1$ is $O(\dimA)$. We apply the reduction again to get from dimension $\dimB'$ to $\dimB$ using $\mathscr P_2$ and $\delta\coloneqq 1 / \N$. This time, we need time $O(\dimB'^2 / \dimB)$ to handle false positives and time $O(\dimB'^2)$ to compute the orthogonal decomposition of $q$. The space overhead of storing an element of $\mathscr P_2$ is $O(\dimB'^2)$. However, these costs are incurred $\dimA / \dimB'$ times, once for each of the subproblems to which the reduction is applied.

We are finally ready to account for the costs of the entire data structure. Let $\C'' = \C(1-\epsA)^2$ be the approximation factor for the $\dimB$-dimensional subproblems. Then the total query time is
\[ O(\dimA\log\dimA) + O\bigp{\frac{\dimA}{\dimB'}} + \frac{\dimA}{\dimB'}\bigp{O\bigp{\frac{\dimB'^2}{\dimB}} + O\bigp{\dimB'^2} + \frac{\dimB'}{\dimB}\cdot\N^{1 / \C''^2 + o(1)}} = \tilde O(\dimA n^{1 / \C^2 + o(1)}), \]
the total space usage is
\[ O(\dimA) + \frac{\dimA}{\dimB'}\bigp{O(\dimB'^2) + \frac{\dimB'}{\dimB}\cdot n^{1 + 1 / \C''^2 + o(1)}} = \tilde O(\dimA n^{1 + 1 / \C^2 + o(1)}), \]
and the total preprocessing time is
\[ \N\cdot O(\dimA\log\dimA) + \frac{\dimA}{\dimB'}\bigp{\N\cdot O(\dimB'^2) + \frac{\dimB'}{\dimB}\cdot n^{1 + 1 / \C''^2 + o(1)}} = \tilde O(\dimA n^{1 + 1 / \C^2 + o(1)}). \]
These match the desired exponent of $\rho = 1 / \C^2 + o(1)$, proving \Cref{thm:main}.

\begin{rem*}
  For the parameter setting $\expC = 4 / 5$ and $\expB = 2 / 5$, the $o(1)$ term in the exponent is in fact bounded by $\tilde O(\log^{-1 / 5}\N)$.
\end{rem*}

\subsubsection{Proof of \Cref{cor:lp}}

We first state a result of \cite{Nguyen14}:

\begin{lem}[{\cite[Theorem 116]{Nguyen14}}]\label{lem:embed}
  There exists a map $f\colon\RR^d\to\RR^{d\poly(\eps^{-1}\log d)}$ such that for all $x, y\in\RR^d$ such that $\norm{x - y}_p\le d^{O(1)}$, we have
  \[ (1 - \eps)\norm{x - y}_p^p - \eps\le\dist{f(x)}{f(y)}^2\le (1 + \eps)\norm{x - y}_p^p + \eps. \]
\end{lem}

To construct a data structure for \Cref{cor:lp} (for $\C$-approximate near neighbors in $\ell_p$) we consider a random shift of the lattice $\dimA\cdot\ZZ^\dimA$ and for each resulting grid cell, build a separate data structure for $\C$-approximate near neighbors on the hypercube of side length $\dimA + 2$ that has the same center as the grid cell. Then, for each point in $\RR^\dimA$, we expect it to belong to $(1 + 2 / \dimA)^\dimA = O(1)$ such data structures, adding a constant factor overhead. To see correctness, observe that there exists a data structure containing every pair of points $x, y\in\RR^\dimA$ such that $\norm{x - y}_p\le 1$.

We finish by constructing a data structure for $\C$-approximate near neighbors in $\ell_p$ for a hypercube of side length $\dimA + 2$. Notice that $\norm{x - y}_p\le O(\dimA\sqrt\dimA)\le d^{O(1)}$ for all $x, y$ in the hypercube. We may thus apply \Cref{lem:embed} with $\eps = 1 / \log\N$ to obtain an embedding into an instance of $\C'$-approximate near neighbors in $\ell_2$, where $\C' = (1 - 2\eps)\C^{p / 2}$. \Cref{thm:main} then gives us a Las Vegas data structure solving $\C$-approximate near neighbors in $\ell_p$ with exponent $\rho = 1 / \C^p + o(1)$.

\section*{Acknowledgments}

I would like to thank Jelani Nelson for advising this project and Huy L\^e Nguy\~{\^e}n for answering questions about \cite{Nguyen14}.

\newpage
\appendix

\section{Las Vegas Data-Dependent Hashing}\label{appendix:B}

Our goal in this appendix is to construct a data-independent Las Vegas filter family for the unit sphere in $\RR^{\dimB}$ and as a consequence of \cite{ALRW17}, which provides optimal space-time tradeoffs for data-dependent hashing, prove the following theorem:

\begin{thm}\label{thm:SB}
  Suppose $\rho_u, \rho_q\ge 0$ are such that
  \[ \C^2\sqrt{\rho_q} + (\C^2 - 1)\sqrt{\rho_u}\ge\sqrt{2\C^2 - 1}. \]
  Then there exists a Las Vegas data structure for Euclidean $\C$-approximate near neighbors in $\RR^\dimB$ with expected query time $O(\N^{\rho_q + o(1)})$, space usage $O(\N^{1 + \rho_u + o(1)})$, and preprocessing time $O(\N^{1 + \rho_u + o(1)})$.
\end{thm}

We accomplish this by applying the high-level approach used in the previous sections to LSF on a sphere. Specifically, we construct a Las Vegas LSF family for the unit sphere in $\RR^\dimC$ by discretizing the sphere to a finite subset for which we have a ``brute force'' probabilistic construction. We then use the ``splitting'' orthogonal decompositions of \Cref{subsec:splitting} to extend this construction to the unit sphere in $\RR^\dimB$. We obtain \Cref{thm:SB} by plugging our Las Vegas LSF family for the unit sphere in $\RR^\dimB$ into the data-dependent data structure of \cite{ALRW17} for $\RR^{\dimB}$. Since the one-sided dimensionality reductions of \Cref{sec:A} still apply, the desired data structure for \Cref{thm:SA} follows.

The resulting construction matches the bounds of \cite{ALRW17} across the entire space-time tradeoff. In particular, for the symmetric case, we match \cite{AR15}, achieving an exponent of $\rho = 1 / (2\C^2 - 1) + o(1)$, which is known to be optimal for data-dependent hashing.

\subsection{Preliminaries}

In this appendix we consider space-time tradeoffs and locality-sensitive filters on a sphere. To facilitate this discussion, we introduce some new definitions, in particular those of asymmetric filter families and spherical LSF.

We also introduce some new notation in this appendix: We use $\Sphere(x, r)$ to denote the sphere of radius $r$ centered at $x$ and $\cN(\mu, \Sigma)$ to denote the normal distribution with mean $\mu$ and covariance matrix $\Sigma$.

\subsubsection{Asymmetric Filter Families}

We introduce \emph{asymmetric filter families}, which are a generalization of (symmetric) filter families to allow for space-time tradeoffs. For the rest of this appendix, we assume our filter families are asymmetric.

\begin{defn}
  An \emph{asymmetric filter family} $\cF$ for a metric space $(X, D)$ is a collection of pairs of subsets of $X$, with each pair consisting of an \emph{update filter} and a \emph{query filter}. We refer to a such a pair collectively as a \emph{filter}. For each $x\in X$, define $\cF_u(x)$ to be the set of filters whose update filter contains $x$ and likewise define $\cF_q(x)$ to be the set of filters whose query filter contains $x$.

  The filter families of \Cref{defn:filterfamily} are simply asymmetric filter families such that $\cF_u = \cF_q$.
\end{defn}

\subsubsection{Spherical LSF}

\emph{Spherical LSF} refers to the class of LSF families for the unit sphere in $d$ dimensions whose filters are defined in terms of inner products with a randomly sampled Gaussian. Previously, this class of LSF families has been studied by \cite{AR15, BDGL16, ALRW17}.

In spherical LSF with parameters $\eta_u$ and $\eta_q$, each filter corresponds to a vector $z\sim\cN(0, I)$, such that a database point $u$ belongs to the filter if $\ip{z, u}\ge\eta_u$ and a query point $q$ belongs to the filter if $\ip{z, q}\ge\eta_q$. To analyze such a filter family, we define
\[ \F(\eta)\coloneqq\PP_{z\sim\cN(0, I)}(\ip{z, x}\ge\eta), \]
where $x\in\cB(0, 1)$ is an arbitrary point, and
\[ \G(s, \eta_u, \eta_q)\coloneqq\PP_{z\sim\cN(0, I)}(\text{$\ip{z, u}\ge\eta_u$ and $\ip{z, q}\ge\eta_q$}), \]
where $u, q\in\Sphere(0, 1)$ are arbitrary points such that $\dist{u}{q} = s$.

The following lemmas give us estimates for $\F$ and $\G$:
\begin{lem}[\cite{ALRW17, AILRS15}]\label{lem:F}
  For $\eta\to\infty$,
  \[ \F(\eta) = \exp\bigp{-(1 + o(1))\frac{\eta^2}{2}}. \]
\end{lem}

\begin{lem}[\cite{ALRW17, AILRS15}]\label{lem:G}
  For $\eta, \sigma\to\infty$ and $0 < s < 2$,
  \[ \G(s, \eta, \sigma) = \exp\bigp{-(1 + o(1))\frac{\eta^2 + \sigma^2 - 2\eta\sigma\Alpha(s)}{2\Beta(s)^2}}, \]
  where $\Alpha(s)\coloneqq 1 - s^2 / 2$ and $\Beta(s)\coloneqq\sqrt{1 - \Alpha^2(s)}$ are the cosine and sine, respectively, of the angle between two points that are distance $s$ apart on the unit sphere.
\end{lem}

\subsection{Spherical LSF Families for $\Sphere(0, 1)$ in $\RR^\dimC$}

We start by constructing a spherical LSF family in dimension $\dimC = \Theta(\log^{\expC}\N)$ with a Las Vegas property, where $\expC = 2 /3$. Specifically, we prove the following:
\begin{prop}\label{prop:filtS}
  For all $r > 0$, there exists a distribution $\Filt$ over spherical LSF families for $\Sphere(0, 1)$ in $\RR^{\dimC}$ with the following properties:
  \begin{enumerate}
    \item\label{filtS:sample}
      A filter family $\cF$ can be sampled from $\Filt$ in time $O(\poly(\dimC^\dimC) + \poly(\dimC) / \G(r, \eta_u, \eta_q))$.
    \item\label{filtS:decode}
      For all $u, q\in\Sphere(0, 1)$, $\cF_u(u)$ and $\cF_q(q)$ can be computed in time $O(\poly(\dimC) / \G(r, \eta_u, \eta_q))$.
    \item\label{filtS:all_pairs}
      For all $u, q\in\Sphere(0, 1)$ such that $\dist{u}{q}\le r$, $\cF_u(u)\cap\cF_q(q)\neq\emptyset$.
    \item\label{filtS:sizes}
      All filters in $\cF$ are subsets of spherical LSF filters sampled with parameters $\eta_u$ and $\eta_q$.
  \end{enumerate}
\end{prop}

We first reduce the problem to spherical caps comparable in size to $r$; our resulting filter family will be the union of the filter families we construct for all the spherical caps we consider. To do so, we take a random shift of the lattice $r\dimC\cdot\ZZ^{\dimC}$ and consider the (axis-aligned) hypercubes with side length $r(\dimC + 2)$ centered at the lattice points. We build a separate filter family for each hypercube. Due to the random shift, each point in $\RR^{\dimC}$ belongs to $(1 + 2 / \dimC)^{\dimC} = O(1)$ hypercubes in expectation. Hence this reduction adds only a constant factor overhead. Furthermore, for each pair of points that are distance at most $r$ apart, there exists a hypercube containing both points, so this step does not separate any pairs of near neighbors.

Now consider a hypercube centered at $O\in r\dimC\cdot\ZZ^\dimC$ that intersects $\Sphere(0, 1)$, and let $O'$ be the point on $\Sphere(0, 1)$ closest to $O$. This hypercube is contained by $\Ball(O', r(\dimC + 2)\sqrt\dimC)$, so it suffices to construct a filter family for the spherical cap $\SphCap = \Sphere(0, 1)\cap\Ball(O', r(\dimC + 2)\sqrt\dimC)$. Let $L$ be a subset of $\SphCap$ of size $\poly((\dimC / \fdel)^{\dimC})$ such that each point in $\SphCap$ is distance at most $r\fdel$ from a point in $L$.\footnote{One construction is to project $r\fdel / \sqrt{\dimC}\cdot\ZZ^{\dimC}\cap\Ball(O', 2r(\dimC + 2)\sqrt{\dimC})$ onto the sphere $\Sphere(0, 1)$.} For $\delta = \Theta(1 / \dimC)$, we can round each point to its nearest neighbor in $L$ with negligible loss in precision---points that are distance $r$ apart become points that are initially distance at most $r(1 + 2\fdel) = r(1 + o(1))$ apart and points that are initially distance $\C r$ apart become points that are distance at least $\C r(1 - o(1))$ apart.

After rounding points to the subset $L$, we show how to construct a spherical LSF family for $\SphCap$. Observe that $p\coloneqq\G(r, \eta_u, \eta_q)$ is a lower bound on the probability that a database point $u$ and a query point $q$ such that $\dist{u}{q}\le r$ both belong to a randomly sampled filter. Thus, if we sample $\fnum\coloneqq\Theta(p^{-1}\log\abs L)$ filters, then with probability $1 / 2$, $\cF_u(u)\cap\cF_q(q)\neq\emptyset$ for all pairs $u, q\in L$ that are distance at most $r$ apart by the union bound. Let the resulting distribution over spherical LSF families for $W$ be $\Filt_0$. Given $\Filt_0$, we can sample a filter family such that the condition holds for all pairs $u, q\in L$ that are distance at most $r$ apart by sampling a filter family from $\Filt_0$ and resampling if the condition does not hold. The checking can be done in time $O(\fnum |L|^2) = O(\poly(\dimC^\dimC))$. Because the success probability is $1/2$, we expect to sample only $O(1)$ times.

This spherical LSF family works for each cap $\SphCap$ associated to a hypercube, since they are all contained in the same-shaped spherical cap. Thus we only need to do the above sampling procedure once to get a spherical LSF family for the entire sphere by taking the union over all caps. Our last step is to analyze the properties of this LSF family and show that it satisfies the properties stated in \Cref{prop:filtS}.

By the definition of sampling above, we immediately have properties \ref{filtS:sample} and \ref{filtS:all_pairs}. Observe that for each cap, there are only $N = O(\poly(\dimC) / \G(r, \eta_u, \eta_q))$ filters. Thus to compute $\cF_u(u)$ (and similarly for $\cF_q(q)$), we can iterate over the set of filters and check for each whether $\ip{z, u}\ge\eta_u$. This gives us property \ref{filtS:decode}. Finally, property \ref{filtS:sizes} follows from the definition of the filter family.\footnote{We actually expect to reject up to $1 / 2$ of the filters sampled this way, so the distribution is slightly different than that of typical spherical LSF, but by the argument for property \ref{filt:local} in \Cref{prop:filt}, this only affects upper bounds on expected values by a factor of $2$.}

\subsection{``Tensoring'' Up Spherical LSF}

We show how to use a modification of the tensoring operation to construct a LSF family for $\Sphere(0, 1)$ in $\RR^{\dimB}$, where $\dimB = \Theta(\log\N\cdot\log\log\N)$.

We start with some parameter settings. Define $\K = \Theta(\log^{1/2}\N)$ and suppose $\rho_u$ and $\rho_q$ are such that
\[ (1 - \Alpha(r)\Alpha(\C r))\sqrt{\rho_q} + (\Alpha(r) - \Alpha(\C r))\sqrt{\rho_u}\ge\Beta(r)\Beta(\C r). \]
Then by \cite[Section 3.3.3]{ALRW17}, there exist $\eta_u$ and $\eta_q$ such that the following hold:
\begin{align*}
  \frac{\F(\eta_u)}{\G(r, \eta_u, \eta_q)}&\le \N^{(\rho_u + o(1)) / \K} \\
  \frac{\F(\eta_q)}{\G(r, \eta_u, \eta_q)}&\le \N^{(\rho_q + o(1)) / \K} \\
  \frac{\G(\C r, \eta_u, \eta_q)}{\G(r, \eta_u, \eta_q)}&\le \N^{(\rho_q - 1 + o(1)) / \K}.
\end{align*}
We can scale $\eta_u$ and $\eta_q$ down by a factor of $\sqrt{\dimB / \dimC}$ to $\eta_u'$ and $\eta_q'$ so that the following hold:
\begin{align*}
  \frac{\F(\eta_u')}{\G(r, \eta_u', \eta_q')}&\le \N^{\frac{1}{\K}\frac{\dimC}{\dimB}(\rho_u + o(1))} \\
  \frac{\F(\eta_q')}{\G(r, \eta_u', \eta_q')}&\le \N^{\frac{1}{\K}\frac{\dimC}{\dimB}(\rho_q + o(1))} \\
  \frac{\G(\C r, \eta_u', \eta_q')}{\G(r, \eta_u', \eta_q')}&\le \N^{\frac{1}{\K}\frac{\dimC}{\dimB}(\rho_q - 1 + o(1))}.
\end{align*}
The spherical LSF parameters $\eta_u'$ and $\eta_q'$ are what we will use to instantiate the distribution $\Filt$ of \Cref{prop:filtS} for our tensoring construction. Finally, define the maximum distortion we allow to be $1\pm\epsB$, where $\epsB = \log^{-\expC / 2}\N$.

We begin our construction by noting that there exists a set $\Proj$ of orthogonal decompositions of $\RR^\dimB$ into $\RR^\dimC$ such that for any pair $u, q\in\Sphere(0, 1)$, there exists a $\cP\in\Proj$ such that every projection in $\cP$ projects $u$, $q$, and $u - q$ with distortion $1\pm\epsB$. Furthermore, we can construct $\Proj$ so that its size is bounded by $\abs{\Proj}\le\poly(\dimB^{\dimB / \dimC})$. This follows from the same CountSketch construction as before, but with a lower tolerance for error so that three distances are preserved instead of one. We further assume without loss of generality that the elements of $\Proj$ are all multiplied by the same random rotation.

Given this collection of orthogonal decompositions $\Proj$, we prove the following:
\begin{prop}\label{prop:filtSB}
  For all $r, \C$ such that $\C > 1$ and $\C r\ge\sqrt 2$, there exists a distribution $\FiltB$ over filter families for $\Sphere(0, 1)$ in $\RR^{\dimC}$ with the following properties:
  \begin{enumerate}
    \item\label{filtSB:sample}
      A filter family $\cG$ can be sampled from $\FiltB$ in time $O(\N^{o(1)})$.
    \item\label{filtSB:decode}
      For all $u, q\in\Sphere(0, 1)$, $\cG_u(u)$ and $\cG_q(q)$ can be computed in time $O(\N^{o(1)} + \abs{\cG_u(u)\cap\cG_q(q)})$. 
    \item\label{filtSB:all_pairs}
      For all $u, q\in\Sphere(0, 1)$ such that $\dist{u}{q}\le r$, $\cG_u(u)\cap\cG_q(q)\neq\emptyset$.
    \item\label{filtSB:sizes}
      For all $u, q\in\Sphere(0, 1)$, we have that
      \[ \EE_{\cG\sim\FiltB}(\abs{\cG_u(u)})\le\poly(\dimB^{\dimB / \dimC})\frac{\F(\eta_u)}{\G(r, \eta_u, \eta_q)}\le\N^{(\rho_u + o(1)) / \K} \]
      and
      \[ \EE_{\cG\sim\FiltB}(\abs{\cG_q(q)})\le\poly(\dimB^{\dimB / \dimC})\frac{\F(\eta_q)}{\G(r, \eta_u, \eta_q)}\le\N^{(\rho_q + o(1)) / \K}. \]
      Furthermore, if $\dist{u}{q} > \C r$, then
      \[ \EE_{\cG\sim\FiltB}(\abs{\cG_u(u)\cap\cG_q(q)})\le\poly(\dimB^{\dimB / \dimC})\frac{\G(\C r, \eta_u, \eta_q)}{\G(r, \eta_u, \eta_q)}\le\N^{(-1 + \rho_u + o(1)) / \K}. \]
  \end{enumerate}
\end{prop}

To construct our LSF family, we independently sample $\dimB / \dimC$ filter families $\cF^{1},\ldots,\cF^{\dimB / \dimC}$ from $\Filt$ with radius $r'\coloneqq r(1 + 8\epsB)$ for each element $\cP\in\Proj$. We then construct a filter family $\cG^{\cP}$ according to a modified version of the tensoring operation: Fix $u\in\Sphere(0, 1)$, and define $u_i'$ to be the projection of $P_i u$ onto $\Sphere(0, \sqrt{\dimC / \dimB})$ in $\RR^{\dimC}$. We then define
\[ \cG^{\cP}_u(u)\coloneqq \cF^{1}_u\bigp{\sqrt{\frac{\dimB}{\dimC}}u_1'}\times\cdots\times\cF^{\dimB / \dimC}_u\bigp{\sqrt{\frac{\dimB}{\dimC}}u_{\dimB / \dimC}'} \]
if $\dist{u_i'}{P_iu}\le\epsB$ for all $i\in [\dimB / \dimC]$ and define $\cG_u^{\cP}$ to be the empty set otherwise. We define $\cG_q^{\cP}(q)$ similarly for query points $q\in\Sphere(0, 1)$. The filter family $\cG$ for $\Sphere(0, 1)$ that we want will be the union of the filter families $\cG^{\cP}$ over all $\cP\in\Proj$.

By the parameter setting of \cite[Section 3.3.3]{ALRW17}, $G(r, \eta_u, \eta_q)\ge\N^{-o(1)}$ when $\C r\ge\sqrt 2$. Thus, by \Cref{prop:filtS}, we can sample from $\Filt$ and compute the filters containing a given point in time $O(\N^{o(1)})$. The overhead of $\Proj$ is also only $\N^{o(1)}$. This gives us properties \ref{filtSB:sample} and \ref{filtSB:decode} for $\cG$.

The ``splitting'' property of $\Proj$ makes our filter family $\cG$ Las Vegas locality-sensitive. For every pair $u,q\in\Sphere(0, 1)$ such that $\dist{u}{q}\le r$, there exists some $\cP\in\Proj$ so that the three distances $u$, $q$, and $u - q$ are preserved (i.e., with distortion within $1\pm\epsB$) for all projections $P\in\cP$. Then $\cG_u^{\cP}(u)$ and $\cG_q^{\cP}(q)$ are non-empty. It follows from the Las Vegas property of filter families drawn from $\Filt_0$ that $u_i'$ and $q_i'$ share a filter for each $i\in [\dimB / \dimC]$, since scaling $P_iu$ and $P_iq$ to $u_i'$ and $q_i'$ increases the distance between them by at most a factor of $1 + \eps$. Hence $\cG_u^{\cP}(u)\cap\cG_q^{\cP}(q)\neq\emptyset$ by \Cref{prop:filtS}, which gives us property \ref{filtSB:all_pairs}.

To bound the number of filters containing a point and the number of collisions, we use the fact that $\Filt$ is a distribution over spherical LSF. Note that for a database point $u$ to belong to a filter characterized by the tuple $z\coloneqq(z_1,\ldots,z_{\dimB / \dimC})$, where the $z_i\in\RR^{\dimC}$ correspond to the centers of the filters in each of the $\dimB / \dimC$ subspaces and $z$ is thought of as belonging to $\RR^\dimB$, we must have
\[ (1\pm\epsB)\ip{z, u} = \sqrt{\frac{\dimC}{\dimB}}\bigp{\sum_{i=1}^{\dimB / \dimC} \ip{z_i, \sqrt{\frac{\dimB}{\dimC}}u_i'}}\ge\sqrt{\frac{\dimB}{\dimC}}\eta_u' = \eta_u. \]
Similarly, we would need $(1\pm\epsB)\ip{z, q}\ge\eta_q$ for a query point $q$. With the random rotation that we applied at the beginning, we get by \Cref{lem:F} that
\[ \PP\bigp{(1 + \epsB)\ip{z, u}\ge\eta_u}\le\F((1 - 2\epsB)\eta_u)\le \N^{o(1) / \K}\F(\eta_u) \]
and
\[ \PP\bigp{(1 + \epsB)\ip{z, q}\ge\eta_q}\le\F((1 - 2\epsB)\eta_q)\le \N^{o(1) / \K}\F(\eta_q). \]
Furthermore, if $\dist{q}{x} > \C r$, then by \Cref{lem:G},
\begin{align*}
  \PP(\text{$(1 + \epsB)\ip{z, u}\ge\eta_u$ and $(1 + \epsB)\ip{z, q}\ge\eta_q$})
  &\le G(\C r, (1 - 2\epsB)\eta_u, (1 - 2\epsB)\eta_q) \\
  &\le\N^{o(1) / \K}\G(\C r, \eta_u, \eta_q).
\end{align*}
Since there are a total of
\[ \abs{\Proj}\cdot\bigp{\frac{\poly(\dimC)}{\G(r, \eta_u', \eta_q')}}^{\dimB / \dimC} = \frac{\poly(\dimB^{\dimB / \dimC})}{\G(r, \eta_u, \eta_q)} \]
filters in a filter family sampled from $\FiltB$ and $\poly(\dimB^{\dimB / \dimC}) = \N^{o(1) / \K}$, property \ref{filtSB:sizes} now follows, completing our proof of \Cref{prop:filtSB}.

Now to prove \Cref{thm:SB}, note that the filter family that we constructed in \Cref{prop:filtSB} can be used as a drop-in replacement for the filter family used in the data-dependent data structure of \cite{ALRW17}. By the parameter settings given in \cite[Section 4.3]{ALRW17}, the filter family will always be instantiated with $\C r\ge\sqrt 2$ (our $\C r$ corresponds to the $r^*$ in \cite{ALRW17}). We now show that the resulting data structure is Las Vegas. The only way randomness leads to false negatives in the data structure of \cite{ALRW17} is through the filtering step, when two near neighbors fail to share a filter. However, this cannot happen with our LSF family by property \ref{filtSB:all_pairs} of \Cref{prop:filtSB}. (Randomness is also used elsewhere, e.g., for the fast preprocessing step described in \cite{AR15}, but this randomness does not create false negatives.)

Finally, to reduce the dimension of the problem from $\dimA$ to $\dimB$, we apply the same dimensionality reductions as in \Cref{sec:A}. To make this reduction step work, we need to slightly modify the data-dependent data structure of \cite{ALRW17}, since our reduction requires the option of continuing after encountering a false positive. In \cite{ALRW17}, at certain nodes of the decision tree, only one point is stored at a vertex as a ``representative'' for specific base cases of the recursion. We require the entire list of points mapped to that vertex to be stored, since it could be that the representative point is a false positive. This modification increases space usage negligibly. We finish by noting that the reduction of \cite{Nguyen14} for $\ell_p$ still applies, from which \Cref{thm:SA} follows.

\section{Las Vegas Locality-Sensitive Filter Families}\label{appendix:A}

In this appendix, we define ``Las Vegas'' data-independent filter families and show that they imply Monte Carlo families of filters as defined in \cite{Christiani17}. Combined with \cite[Theorem 1.5]{Christiani17}, this reduction implies analogs of the data-independent lower bounds of \cite{OWZ14} hold for Las Vegas LSF families. It follows from these lower bounds that the exponents of \Cref{thm:main} and \Cref{cor:lp} are optimal for LSF in the data-independent setting.

Let $(X, D)$ be a metric space. We start with a Monte Carlo definition of LSF families as considered in \cite{Christiani17}:

\begin{defn}[(Symmetric\footnote{In \cite{Christiani17}, asymmetric filter families yielding space-time tradeoffs for LSF are also discussed. These filter families have separate ``query'' and ``update'' filters. However, since this section focuses only on the symmetric setting, we merge the (asymmetric) LSF parameters $p_q$ and $p_u$ of \cite{Christiani17} into the single (symmetric) parameter $q$.}) Monte Carlo filter families \cite{Christiani17}]\label{defn:MCfilt}
  A distribution $\mathscr{F}$ over filters is said to be \emph{Monte Carlo $(r, cr, p_1, p_2, q)$-sensitive} if the following holds:
  \begin{itemize}
    \item
      For all $x, y\in X$ such that $D(x, y)\le r$, $\PP_{F\sim\mathscr{F}}(x\in F, y\in F)\ge p_1$.
    \item
      For all $x, y\in X$ such that $D(x, y) > \C r$, $\PP_{F\sim\mathscr{F}}(x\in F, y\in F)\le p_2$.
    \item
      For all $x\in X$, $\PP_{F\sim\mathscr{F}}(x\in F)\le q$.
  \end{itemize}
\end{defn}

To obtain Las Vegas algorithms using LSF, we need a stronger notion of filter family, which we define as follows:

\begin{defn}[Las Vegas filter families]
  A distribution $\Filt$ over filter families is said to be \emph{Las Vegas $(r, \C r, p_2, q, N)$-sensitive} if the following holds:
  \begin{itemize}
    \item
      For all $x, y\in X$ such that $D(x, y)\le r$, $\PP_{\cF\sim\Filt}(\cF(x)\cap\cF(y)\neq\emptyset) = 1$.
    \item
      For all $x, y\in X$ such that $D(x, y) > \C r$, $\EE_{\cF\sim\Filt}(\abs{\cF(x)\cap\cF(y)})\le p_2$.
    \item
      For all $x\in X$, $\EE_{\cF\sim\Filt}(\abs{\cF(x)})\le q$.
    \item
      The number of filters in $\cF$ sampled from $\Filt$ is always $N$.
  \end{itemize}
\end{defn}

This definition is general in the sense that for a set of filters to be usable in a Las Vegas setting, it must guarantee that every pair of ``close'' points is contained in a filter. Furthermore, $p_2$ and $q$ are analogous to those of \Cref{defn:MCfilt}, and our lower bound is independent of the size $N$.

With this definition of LSF in a Las Vegas setting, the remaining proofs are straightforward: We first show a general reduction from Las Vegas LSF to Monte Carlo LSF, and then, for the specific case of $\ell_p$ spaces, we show that the lower bound of \cite{Christiani17} and \cite{OWZ14} still holds.

\begin{lem}\label{lem:LStoMC}
  Suppose a distribution $\Filt$ over filter families is Las Vegas $(r, \C r, p_2, q, N)$-sensitive, and define $p_1'\coloneqq 1 / N$, $p_2'\coloneqq p_2 / N$, and $q'\coloneqq q / N$. Then there exists a distribution $\mathscr{F}$ over filters that is Monte Carlo $(r, \C r, p_1', p_2', q')$-sensitive.
\end{lem}

\begin{proof}
  Let $\mathscr{F}$ be the distribution over filters where we first sample $\cF\sim\Filt$ and then sample a filter uniformly at random from $\cF$. One can check that $\mathscr{F}$ is Monte Carlo $(r, \C r, p_1', p_2', q')$-sensitive.
\end{proof}

\begin{prop}
  Suppose $0 < p < \infty$. Every Las Vegas $(r, \C r, p_2, q, N)$-sensitive filter family $\Filt$ for $\RR^\dimA$ under the $\ell_p$ norm must satisfy
  \[ \rho\coloneqq\frac{\log(q)}{\log(q / p_2)}\ge \frac{1}{\C^p} - o_d(1) \]
  when $p_2 / q \ge 2^{-o(\dimA)}$.
  (In particular, this lower bound depends only on $p_2$ and $q$ and not on $N$.)
\end{prop}

\begin{proof}
  By the reduction of \Cref{lem:LStoMC} and \cite[Theorem 1.4]{Christiani17}, we have that
  \[ \rho = \frac{\log(q' / p_1')}{\log(q' / p_2')}\ge\frac{1}{\C^p} - o_d(1). \]
\end{proof}

\newpage
\bibliographystyle{alpha}
\bibliography{writeup}

\end{document}